%% file: digraphs.tex
\documentclass{llncs}
\pdfoutput=1
\usepackage[english, ngerman]{babel}
\usepackage[latin9]{inputenc}
\usepackage[T1]{fontenc}
\usepackage{amsmath}
\usepackage{amssymb}
\usepackage{mathrsfs}
\usepackage{algorithmic}
\usepackage{algorithm}
\usepackage{fancybox}
\usepackage{tikz}
\usepackage{float}
\usepackage[format=plain,font=it,labelfont=bf]{caption}
\usetikzlibrary{decorations.pathmorphing}
\usetikzlibrary{shapes}

\title{Property-Testing in Sparse Directed Graphs: $3$-Star-Freeness and Connectivity}
\author{Frank Hellweg and Christian Sohler\thanks{Research partly supported by DFG grant SO 514/3-2 and ERC Starting Grant 307696}}
\institute{Department of Computer Science, Technische Universit\"{a}t Dortmund \\\email{\{frank.hellweg,christian.sohler\}@tu-dortmund.de}}

\input{Makros}

\begin{document}
\selectlanguage{english}
\maketitle

\begin{abstract}We study property testing in directed graphs in the bounded degree model, where we assume that an algorithm may only query the outgoing edges of a vertex, a model proposed by Bender and Ron \cite{BR02}. As our first main result, we we present a property testing algorithm for strong connectivity in this model, having a query complexity of $\Ovon(n^{1-\epsilon/(3+\alpha)})$ for arbitrary $\alpha>0$; it is based on a reduction to estimating the vertex indegree distribution. For subgraph-freeness we give a property testing algorithm with a query complexity of $\Ovon(n^{1-1/k})$, where $k$ is the number of connected componentes in the queried subgraph which have no incoming edge. We furthermore take a look at the problem of testing whether a weakly connected graph contains vertices with a degree of least $3$, which can be viewed as testing for freeness of all orientations of $3$-stars; as our second main result, we show that this property can be tested with a query complexity of $\Ovon(\sqrt{n})$ instead of, what would be expected, $\Omega(n^{2/3})$.
\end{abstract}

\section{Introduction}

Property testing is a technique for solving decision problems that sacrifices some accuracy for the benefit of a sublinear time complexity. The sacrifice of accuracy is twofold: On the one hand, we allow property testing algorithms to accept a small margin of inputs that do not have the queried property $\Pi$ but are similar to some inputs that have $\Pi$. More formally, for a proximity parameter $\epsilon<1$, we say that an input is \emph{$\epsilon$-far} from having the property $\Pi$, if one must modify an $\epsilon$-fraction of the input's description in order to construct an input that has $\Pi$. We only require a property testing algorithm for $\Pi$ to give a reliable answer for inputs that either have the property $\Pi$ or are $\epsilon$-far from it.

The second relaxation in accuracy 
is due to the randomized nature of property testing algorithms: All those algorithms are \emph{Monte Carlo} algorithms, which means that they are allowed to have a small constant error probability. 

The most important measure for the performance of a property testing algorithm is its \emph{query complexity}, which is the worst-case number of accesses to the input that it needs for inputs of a given size.
We aim for algorithms that have a query complexity of $\ovon(n)$ or even $\Ovon(1)$.

In this paper we are particularly interested in property testing for sparse directed graphs. Such graphs are assumed to be stored in adjacency list representation and have both an in- and an outdegree of at most some constant $d$; we require the adjacency lists to only contain the outgoing edges of a vertex, a model which has been introduced in \cite{BR02}. This is a quite natural model for directed graphs: For example, the webgraph or, typically, graphs of social networks are sparse graphs which have directed links; in particular, the incoming edges of a vertex of these graphs might not be visible, for example in case of the incoming links of a website during a web crawl. To gain this knowledge, basically the whole graph has to be explored, and since these graphs are typically very large, this may be inappropriate. Property testing algorithms for this graph model can be useful to gain information about the structure of such graphs while exploring only a small portion of it.

Property testing has been introduced by Rubinfeld and Sudan \cite{RS96}, while Goldreich, Goldwasser, and Ron \cite{GGR98} have initiated the study of graph properties. 
In this paper the authors introduced property testing in the dense graph model, where graphs are assumed to be stored as an adjacency matrix. Furthermore, Goldreich and Ron have introduced property testing in the sparse graph model \cite{GR97}. Since then a large variety of graph properties has been studied, including \cite{BSS08,CSS09,HKNO09,NS11} in the
sparse graph model and \cite{AFNS06} in the dense graph model. These papers aim for identifying classes of testable properties: For the sparse graph model, the above series of papers shows that every hyperfinite graph property is testable, as well as every property in hyperfinite graphs; in the dense graph model, a graph property is testable if and only if it can be reduced to a problem of testing for satisfaction of one of a finite number of Szemer\'{e}di-Partitions.

Property testing in directed graphs can also be subdivided into property testing in the dense graph and the sparse graph models. In the dense graph model, Alon and Shapira have studied the property of subgraph-freeness \cite{AS03}. Bender and Ron have studied the property of acyclicity in both the sparse graphs and the dense graph model and the property of strong connectivity for sparse graphs \cite{BR02}. In the sparse graph model, they show that if a property testing algorithm is only allowed to query the outgoing edges of a vertex, there are no such algorithms with a query complexity of $o(n^{1/3})$ for acyclicity and $o(n^{1/2})$ for strong connectivity, where $n$ is the number of vertices of the input graph. The assumption that only the outgoing edges of a vertex may be queried makes testing strong connectivity much harder: As Bender and Ron show, there is a one-sided error property testing algorithm with a query complexity of $\tilde\Ovon(1/\epsilon)$ for strong connectivity without this constraint. Finally, Yoshida and Ito give a constant-time property testing algorithm for $k$-edge connectivity of directed graphs \cite{YI10}, which also relies on the visibility of incoming edges.

\paragraph{Our Results.}In this paper we further study property testing in sparse directed graphs where only the outgoing edges of a vertex may be queried. The first property we study is subgraph-freeness, i.e., to test whether a graph $H$ does not occur as a subgraph of a graph $G$. Let $k$ be the number of connected components of $H$ that have no incoming edge from another part of $H$: Then our algorithm has a query complexity of $\Ovon(n^{1-1/k})$.

A problem connected to subgraph-freeness is testing whether a weakly connected graph is free of all orientations of $3$-stars. Birthday-paradox type arguments would imply a query complexity of $\Omega(n^{2/3})$ for this problem, but we can give an algorithm with one of $\Ovon(n^{1/2})$, which is the first main result of the paper. This algorithm makes use of two facts: The first is that the above mentioned class of forbidden subgraphs induces some strong properties for graphs that are free of them; the second is that, when sampling edges, the probability of hitting a vertex twice as the target vertex of two different edges is disproportionally high if it has many incoming edges. This allows the algorithm to compute a ratio of two estimators, which will be considerably larger if the input graph has many vertices with a degree of at least $3$.

The second main result of this paper is a property testing algorithm for strong connectivity that achieves a query complexity of $\Ovon(n^{1-\epsilon/(3+\alpha)})$ for arbitrary $\alpha>0$. The algorithm is based on a reduction of the strong connectivity problem to a problem of estimating the vertex indegrees of a graph: We show that it is possible to define a locally computable partitioning of the input graph, such that small connected components that have no incoming edges become their own partitions; one can then construct a metagraph in which every partition of the original graph becomes a vertex. If the input graph is far from strongly connected, then the metagraph will contain many vertices with an indegree of $0$, which can indeed be tested by statistics of the vertex indegrees. After first publishing this result at ESA 2012 \cite{HS12} we discovered that there already existed a proof sketch for an algorithm similar to ours, which Oded Goldreich published in the appendix of a survey article about graph property testing \cite{G10}. At the end of the corresponding section we include a detailed discussion about similarities and differences between Goldreich's and our algorithm.

\section{Preliminaries}
The graph model studied in this paper is the sparse graph model. If not explicitly stated else, all graphs in this paper are directed graphs whose vertices have an outdegree which is bounded by a constant $d$, as well as the indegree; this follows the notion in \cite{BR02}. The graphs are assumed to be stored as adjacency lists. We at first define the notion of $\epsilon$-farness:

\begin{definition}
Let $G$, $H$ be directed graphs as above, both having $n$ vertices. We say that $G$ is $\epsilon$-far from $H$, if one has to change more than $\epsilon dn$ entries of the adjacency lists of $G$ to obtain a graph that is isomorphic to $H$.

Let $\Pi$ be a graph property. We say that $G$ is $\epsilon$-far from $\Pi$, if it is $\epsilon$-far from any graph in $\Pi$.
\end{definition}
Note that graphs as defined above have at most $dn$ edges. This implies that changing $\epsilon dn$ entries of adjacency lists means changing an $\epsilon$-fraction of the graph's description. We can now define the way property testing algorithms get access to an input graph:

\begin{definition}Let $G=(V,E)$ be a directed graph with each vertex having an outdegree of at most $d\in\NN$. We define $f_G:V\times \NN\rightarrow V\cup\{+\}$ to be a function that for querying $f(v,i)$ returns the $i$-th neighbour of vertex $v\in V$ in the adjacency list representation of $G$, or $+$, if $v$ has less than $i$ neighbours.
\end{definition}
Property testing algorithms get access to $f_G$ to gain knowledge about the input graph. A call to $f_G$ takes $\Ovon(1)$ time.

\begin{definition}
Let $\calA$ be an algorithm that has parameters $f_G$, $\epsilon$ and $n$. We define the \emph{query complexity} of $\calA$ as the worst case number of calls to $f_G$ it performs for any graph $G$ with $n$ vertices. $\calA$ is a property testing algorithm for a graph property $\Pi$, if:
\vspace{-.1cm}
\begin{enumerate}
	\item The query complexity of $\calA$ is sublinear in $n$.
	\item $\calA$ accepts every graph $G\in\Pi$ with a probability of at least $\frac{2}{3}$.
	\item $\calA$ rejects every graph $G$ that is $\epsilon$-far from $\Pi$ with a probability of at least $\frac{2}{3}$.
\end{enumerate}
If $\calA$ accepts every $G\in\Pi$ with probability $1$, we say it has $1$-sided error, else we say it has $2$-sided error.
\end{definition}
Finally, we define some graph properties that we will need throughout the rest of this paper. Let $G=(V,E)$ and $H=(V',E')$ be directed graphs.

We call $H$ a \emph{subgraph} of $G$, if there exists an injective mapping $g:V'\rightarrow V$ such that $(g(u),g(v))\in E$ for all $(u,v)\in E'$; we say that $G$ is \emph{$H$-free}, if $H$ is not a subgraph of $G$.

We call $G$ \emph{(strongly) connected}, if for all pairs of vertices $u,v\in V$ there is a directed path from $u$ to $v$ in $G$ (we also say that $v$ can be \emph{reached} from $u$). We call $G$ \emph{weakly connected}, if for all $u,v\in V$ there is an undirected path between $u$ and $v$. $U\subseteq V$ is a connected component of $G$, if the subgraph of $G$ induced by $U$ is strongly connected and there is no set of vertices $W\subseteq V-U$ such that the subgraph of $G$ induced by $U\cup W$ is strongly connected; i.e., $U$ is maximal.

We have to distinguish between several types of connected components of a graph that are witnesses to it being not strongly connected: Source and sink components:

\begin{definition}[source and sink components and dead ends]
Let $G=(V,E)$ be a directed graph. A strongly connected compontent $U\subseteq V$ is called a \emph{source component}, if there is no edge from $V-U$ to $U$; $U$ is called \emph{sink component}, if there is no vertex from $U$ to $V-U$. We will call either of those components \emph{dead ends}.
\end{definition}

We next define a special type of undirected graphs, $k$-stars. Such graphs consist of a central vertex that is connected to $k$ other vertices:

\begin{definition}[$k$-star]
	An undirected graph $H=(V,E)$ is called \emph{$k$-star}, if the following holds:
	\begin{itemize}
		\item $H$ has $k+1$ vertices;
		\item there is a vertex $v\in V$, such that for every vertex $u\in V-\{v\}$ there is an edge $\{u,v\}$;
		\item $G$ does not contain any other edges.
	\end{itemize}
\end{definition}

In a directed setting, we will consider orientations of $k$-stars. If a $k$-star orientation occurs as a subgraph of a graph $G$, then we call the central vertex of this occurence \emph{$k$-star vertex}. We call a vertex \emph{incoming} $k$-star vertex if it has at least $k$ incoming edges and we call it \emph{outgoing} $k$-star vertex if it has at least $k$ outgoing edges.

To simplify the analysis of our algorithms, we use a sampling technique that deviates from the usual sampling of vertices (respectively, edges) with replacement. Instead, we sample each vertex (edge) of the input graph with a certain probability $p$. If a fixed number of vertex (edge) samples is exceeded, the algorithm aborts by returning an arbitrary answer; the probability of this event will be small. Note that in our analyses the case that the sample limit is exceeded seperately is considered seperately from the rest of the particular analysis; after that we use a union bound to bound the total error. Thus, in the rest of the analyses we can assume that each vertex (edge) is independently sampled with probability $p$.

\input{Kap_Subgraphs}
\input{Kap_Connectivity}

\input{bibfile}

\end{document}

%% file: Makros.tex
\newtheorem{mtheorem}{Theorem}
\newtheorem{observation}{Observation}

\newcommand{\E}{\ensuremath{\mathrm{E}}}

\newcommand{\Ovon}{\ensuremath{\mathcal{O}}}
\newcommand{\ovon}{\ensuremath{\textit{o}}}

\newcommand{\NN}{\mathbb{N}}

\newcommand{\calA}{\mathcal{A}}
\newcommand{\calB}{\mathcal{B}}
\newcommand{\calC}{\mathcal{C}}

\newcommand{\calH}{\mathcal{H}}
\newcommand{\degin}{\deg_{\text{in}}}
\newcommand{\degout}{\deg_{\text{out}}}
\newcommand{\indeg}{\deg_{\text{in}}}

\floatstyle{plain}
\newfloat{Algorithm}{tbh}{lop}

\definecolor{gray70} {gray} {0.7}
\definecolor{darkgreen}{rgb}{0.0,0.6,0.0}
\definecolor{Apricot}     {cmyk}{0,0.32,0.52,0}
\definecolor{Aquamarine}  {cmyk}{0.82,0,0.30,0}
\definecolor{CadetBlue}   {cmyk}{0.62,0.57,0.23,0}
\definecolor{DarkGray}    {gray}{0.2}
\definecolor{DarkGreen}   {rgb}{0,0.5,0}
\definecolor{ForestGreen} {cmyk}{0.91,0,0.88,0.12}
\definecolor{Gold}        {rgb}{1.,0.84,0.}
\definecolor{Goldenrod}   {cmyk}{0,0.10,0.84,0}
\definecolor{IndianRed}   {rgb}{0.8,0.36,0.36}
\definecolor{Lavender}    {cmyk}{0,0.48,0,0}
\definecolor{LemonChiffon}{rgb}{1.,0.98,0.8}
\definecolor{LightBlue}   {rgb}{0.68,0.85,0.9}
\definecolor{LightCyan}   {rgb}{0.88,1.,1.}
\definecolor{LightGray}   {gray}{0.92}
\definecolor{LightYellow} {rgb}{1.,1.,0.88}
\definecolor{Melon}       {cmyk}{0,0.46,0.50,0}
\definecolor{NavyBlue}    {cmyk}{0.94,0.54,0,0}
\definecolor{Orange}      {rgb}{1.,0.65,0.}
\definecolor{PaleGreen}   {rgb}{0.6,0.98,0.6}
\definecolor{PaleGreenB}  {rgb}{0.9,1,0.9}
\definecolor{Peach}       {cmyk}{0,0.50,0.70,0}
\definecolor{PeachPuff}   {rgb}{1.0,0.85,0.73}
\definecolor{PineGreen}   {cmyk}{0.92,0,0.59,0.25}
\definecolor{Pink}        {rgb}{1.,0.75,0.8}
\definecolor{RoyalBlue}   {cmyk}{1,0.50,0,0}
\definecolor{SeaGreen}    {cmyk}{0.69,0,0.50,0}
\definecolor{Salmon}      {cmyk}{0,0.53,0.38,0}
\definecolor{Sepia}       {cmyk}{0,0.83,1,0.70}
\definecolor{SlateBlue}   {rgb}{0.42,0.35,0.8}
\definecolor{Thistle}     {rgb}{0.85,0.75,0.85}
\definecolor{Turquoise}   {cmyk}{0.85,0,0.20,0}
\definecolor{Violet}      {cmyk}{0.79,0.88,0,0}
\definecolor{YellowOrange}{cmyk}{0,0.42,1,0}



%% file: Kap_Subgraphs.tex
\section{Testing $3$-Star-Freeness}\label{kap_subgraphs}

We start by developing a simple property testing algorithm with one sided error for a very basic graph property, \emph{subgraph freeeness}.

\begin{definition}
Let $G$ and $H$ be directed graphs. We call $G$ (induced) $H$-free if $H$ does not appear as an (induced) subgraph of $G$.
\end{definition}

This algorithm has one-sided error and a query complexity of $\Theta(n^{1-1/k})$, if the forbidden subgraph $H$ has $k$ source components. We will later use this algorithm as a subroutine for a more complex algorithm that tests for freeness of a certain class of subgraphs: The class of all orientations of $3$-Stars. By running the simple property tester for subgraph freeness for every possible $3$-star orientation, one would achieve a query complexity of $\Ovon(n^{2/3})$. This is because there is a $3$-star orientation where the central vertex has $3$ incoming edges and thus this graph has $3$ source components.

We give a more refined algorithm that requires the input graph to be weakly connected and that uses statistical measures to distinguish between graphs that have many occurences of such $3$-star orientations with only incoming edges and graphs that are $3$-star-free; these statistical measures include $2$-way collision-statistics on the target vertices of samples edges, which can be done with $\Ovon(n^{1/2})$ queries. For all other types of $3$-star orientations occuring in the input graph, the simple subgraph freeness algorithm is called, each of these calls also having a query complexity of $\Ovon(n^{1/2})$. Hence, our algorithm has a query complexity of $\Ovon(n^{1/2})$ (considering $\epsilon$ as a constant) and thus breaks the trivial barrier of $\Ovon(n^{2/3})$ for this problem. To our knowledge, this is the first algorithm that does so for a nontrivial property of directed graphs.

At the end of the chapter we show that this algorithm is asymptotically optimal for testing $3$-star freeness: We show that any property testing algorithm for this problem has a query complexity of $\Omega(n^{1/2})$, and that even $\Omega(n^{2/3})$ queries are required if the input graph is allowed to be disconnected.

\subsection{A Property Testing Algorithmus for $H$-Freeness}

\begin{Algorithm}[tb]
\noindent\centering
\shadowbox{
	\begin{minipage}{8cm}
		\begin{tabbing}
			~~~~\= ~~~~\= ~~~~\= ~~~~ \= \kill
			\textsc{\textbf{TestSubgraphFreeness}}($G,H,\epsilon$)\\
			\> Sample every vertex of $G$ with probability $p=\left(\frac{6m}{\epsilon n} \right)^{1/k}$;\\
			\> \> \> if more than $4\cdot \left(\frac{6m}{\epsilon} \right)^{1/k}n^{1-1/k}$ are sampled\textbf{then return true}\\
			\> \textbf{Foreach} sampled vertex $v$ {\bf do}\\
			\> \> Start a breadth first search at $v$ having a maximum depth of $m$\\
			\> \> \textbf{If} BFS explores an occurence of $H$ in $G$ or BFS completes\\
			\> \> \> \> an occurence that was partly explored before \textbf{then return false}\\
			\> \textbf{return true}
		\end{tabbing}
	\end{minipage}
}
\caption{{\sc TestSubgraphFreeness}}\label{alg_subgraphfreiheit}
\end{Algorithm}

The algorithm \textsc{TestSubgraphFreeness} is given to directed bounded-degree graphs $G$ und $H$, where $G$ has $n$ vertices and $H$ has $m$ vertices. The algorithm is also given a proximity parameter $\epsilon$. Let $k$ be the number of source components of $H$.

In order to find an occurence of $H$ in $G$, \textsc{TestSubgraphFreiheit} samples every vertex of $G$ with a probability of $p=\left(\frac{6m}{\epsilon n} \right)^{1/k}$, such that the expected number of sampled vertices is $np=\Ovon(\epsilon^{-1}m^{1/k}n^{1-1/k})$. We will see that this is sufficient for with high probability sampling a vertex in every source component of at leastone occurence of $H$, if $G$ is $\epsilon$-far from $H$-free. The maximum depth of $m$ of the breadth first searches started in the sampled vertices ensures that this suffices to completely explore that occurence of $H$.

We note that {\sc TestSubgraphFreeness} can only reject the input, if it actually explores a occurence of $H$; hence, if $G$ is $H$-free, it can never be rejected, and thus {\sc TestSubgraphFreeness} has a one-sided error. It remains to bound the probability for inadvertantly accepting the input, if $G$ is $\epsilon$-far from $H$-free. The following lemma bounds the number of vertex-disjoint occurences of $H$ in $G$ from below in this case:

\begin{lemma}\label{lemma_subgraphcount}
Let $G=(V,E)$ and $H=(V_H,E_H)$ be directed graphs with both the vertex indegrees and outdegrees bounded by $D\in\NN$; let $\epsilon<1$ be a proximity parameter and assume that $G$ is $\epsilon$-far from $H$-free. Let also $|V|=:n$ and $|V_H|=:m$. Then $G$ contains at least $\frac{\epsilon n}{2m}$ vertex-disjoint copies of $H$.
\end{lemma}

\begin{proof}
Let $H$ consist of $l$ weak connected components $H_1,\ldots, H_l$. We can assume that each $H_i$ contains at least two vertices, and, because it is weakly connected, every vertex of $H_i$ is incident to an edge. We assume that $G$ contains less than $\frac{\epsilon n}{2m}$ vertex-disjoint occurences of $H$ lead this to a contradiction.

Let $M$ be any inclusionwise maximal set of vertex-disjoint occurences of $H$ in $G$. Then there is at least one connected componentn $H_i$ of $H$, for which all the remaining occurences are not vertex-disjoint with the occurences of $H$ in $M$. Thus, at least one edge of each of these occurences of $H_i$ is incident to a vertex that belongs to some occurence of $H$ in $M$.

Now consider the graph $G'$ that results from removing any edges that are incident to vertices of occurences in $M$: $G'$ does not contain an occurence of $H_i$ any more, and thus does not contain a copy of $H$ either. Thus, $G'$ is $H$-free. On the other hand, $M$ has a cardinality of less than $\frac{\epsilon n}{2m}$, each occurence of $H$ has $m$ vertices, and the maximum number of edges incident to any of these vertices is $2D$. Thus, the number of edges that are deleted from $G$ to get $G'$ is less than $2Dm\cdot\frac{\epsilon n}{2m}=\epsilon Dn$. This is a contradiction to the assumption that $G$ is $\epsilon$-far from $H$-free.
\end{proof}

It remains to show that the size of the vertex sample in {\sc TestSubgraphFreeness} is sufficiently large for, with high probability, exploring at least one occurence of $H$, if $G$ is $\epsilon$-far from $H$-free. We at first assume that $H$ is weakly connected.

\begin{mtheorem}\label{theorem_subgraphtest}
Let $G=(V,E)$ and $H=(V_H,E_H)$ be directed graphs with both the vertex indegrees and outdegrees bounded by $D\in\NN$ and let $H$ be weakly connected; let $\epsilon<1$ be a proximity parameter and let $k$ be the number of source components of $H$. Let also $|V|=:n$ and $|V_H|=:m$.

Then {\sc TestSubgraphFreeness$(G,H,\epsilon)$} returns \emph{true}, if $G$ is $H$-free and returns \emph{false} with probability at least $2/3$, is $G$ is $\epsilon$-far from $H$-free. {\sc TestSubgraphFreeness} has a query complexity of $\Ovon\left( D^m\left(\frac{m}{\epsilon} \right)^{1/k}n^{1-1/k} \right)$.
\end{mtheorem}

\begin{proof}
As discussed above,  {\sc TestSubgraphFreeness} can only reject if $G$ contains at least one occurence of $H$; thus, the algorithm always accepts if $G$ is $H$-free.

Now assume that $G$ is $\epsilon$-far from $H$-free and let $M$ be a maximal set of occurences of $H$ in $G$; Lemma \ref{lemma_subgraphcount} ensures that there is a set $M$ with a cardinality of at least $\epsilon n/2m$. Let $X_i$ be the event that the $i$-th occurence in $M$ (according to an arbitrary but fixed ordering of the occurences in $M$) is explored completely by {\sc TestSubgraphFreeness}. Since the BFS traversals have a depth of $m$, this happens if a vertex in each of the source components of this occurence is sampled; since there are $k$ of those, it holds $\Pr[X_i]\geq p^k$ for all $i$.

Since all the occurences in $M$ are vertex-disjoint (and due to the sample process), the events $X_i$ are independent. Hence,

\begin{align*}
	\Pr\left[\bigcap_{1\leq i \leq |M|}\bar X_i\right] & = \prod_{1\leq i \leq |M|}\left(1-\Pr[X_i]\right) \leq \left(1-p^k\right)^{\epsilon n/2m}\\
		&  = \left( 1-\frac{6m}{\epsilon n} \right)^{\epsilon n/2m}\leq e^{-3} < \frac{1}{12}
\end{align*}

gives an upper bound for the probability that none of the occurences in $M$ is explored completely and thus an upper bound for the probability that none of the occurences of $H$ in $G$ is identified.

It remains to bound the probability that the input is inadvertantly accepted in the first line of {\sc TestSubgraphFreeness} due to too many vertices being sampled. The expected number of vertices sampled is $np=\left(3m/\epsilon\right) ^{1/k}\cdot n^{1-1/k}$, and, due to the Markov Inequality, the probability, that more than four times this number of vertices is sampled is at most $1/4$. Together with the above bound on the probability of not exploring a single occurence of $H$, the union bound gives a total probability of at most
$1/4+1/12 = 1/3$ for inatvertantly accepting the input.

The number of vertices sampled by {\sc TestSubgraphFreeness} is at most $4\cdot \left(\frac{3m}{\epsilon} \right)^{1/k}n^{1-1/k}+1$, for at most $4\cdot \left(\frac{3m}{\epsilon} \right)^{1/k}n^{1-1/k}$ of those a BFS traversal of depth $m$ is started. Since the vertex outdegree of $G$ is bounded by $D$, the number of edges queried in each of the BFS traversals is bounded by $D^m$, and thus the total number of queries needed is $\Ovon\left( D^m\left(\frac{m}{\epsilon} \right)^{1/k}n^{1-1/k} \right)$.
\end{proof}

Now we assume that $H$ is disconnected, i.e. consists of $l>1$ weakly connected components $H_1,\ldots ,H_l$. It suffices to identify occurences of each of the $H_i$ independently, all of those pairwise vertex-disjoint. In a maximal set of vertex-disjoint occurences of $H$, as guaranteed by lemma \ref{lemma_subgraphcount}, the occurences of the individual $H_i$ can be combined arbitrarily to get an occurence of $H$.

Thus, {\sc TestSubgraphFreeness} can be run individually for each of the $H_i$. For this purpose, we define the algorithm {\sc TestSubgraphFreenessAmplified}, which uses probability amplification to guarantee a better success probability $p$ instead of $2/3$.

\begin{Algorithm}[tb]
\noindent\centering
\shadowbox{
	\begin{minipage}{8cm}
		\begin{tabbing}
			~~~~\= ~~~~\= ~~~~\= ~~~~ \= \kill
			\textsc{\textbf{TestSubgraphFreenessAmplified}}($G,H,\epsilon,p$)\\
			\> \textbf{For} $i\rightarrow 1$ \textbf{to} $\log_3 \left\lceil \frac{1}{p}\right\rceil$\\
			\> \> \textbf{If not} {\sc TestSubgraphFreeness}($G,H,\epsilon,p$) \textbf{then return false}\\
			\> \textbf{return true}
		\end{tabbing}
	\end{minipage}
}
\caption{{\sc TestSubgraphFreenessAmplified}}
\end{Algorithm}

\begin{lemma}\label{lemma_subgraphtestverstärkt}
Let $G=(V,E)$ and $H=(V_H,E_H)$ be directed graphs with both the vertex indegrees and outdegrees bounded by $D\in\NN$; let $\epsilon<1$ be a proximity parameter and $p\leq\frac{1}{3}$. Then {\sc TestSubgraphFreenessAmplified} returns \emph{true} if $G$ is $H$-free and \emph{false} with probability at least $1-p$ if $G$ is $\epsilon$-far from $H$-free. {\sc TestSubgraphFreenessAmplified} has a query complexity of $\Ovon\left(\log\frac{1}{p} \cdot D^m\left(\frac{m}{\epsilon} \right)^{1/k}n^{1-1/k} \right)$.
\end{lemma}

\begin{proof}
In case $G$ is $H$-free, the correctness follows directly from the correctness of {\sc TestSubgraphFreeness}. Now assume that $G$ is $\epsilon$-far from $H$-free. Theorem \ref{theorem_subgraphtest} guarantees that every call of {\sc TestSubgraphFreeness} returns \emph{false} with probability at least $2/3$. Thus, the probability that none of the calls returns \emph{false} (and thus {\sc TestSubgraphFreenessAmplified} inadvertantly returns true) is at most $(1/3)^{\log_3 \left\lceil 1/p \right\rceil} \leq 3^{-\log_3 1/p} = p$.

Since {\sc TestSubgraphFreeness} is called at most $\log_3 \left\lceil \frac{1}{p}\right\rceil=\Ovon\left(\log \frac{1}{p}\right)$ times in {\sc TestSubgraphFreenessAmplified}, the query complexity of the latter algorithm is $\Ovon\left(\log\frac{1}{p} \cdot D^m\left(\frac{m}{\epsilon} \right)^{1/k}n^{1-1/k} \right)$.
\end{proof}

For testing $H$-freeness for disconnected $H$, we run {\sc TestSubgraphFreenessAmplified} individually for each of the $l$ connected components of $H$ while setting the parameter $p$ to $\frac{1}{3l}$. Assume that $G$ is $\epsilon$-far from $H$-free: Then the probability that at least one of the $l$ calls to {\sc TestSubgraphFreenessAmplified} fails to identify a complete occurence of the corresponding component is at most $l\cdot p=1/3$ due to the union bound. Hence we get the following corollary:

\begin{corollary}
$H=(V_H,E_H)$ be a directed graph with $m$ vertices and with both the vertex indegrees and outdegrees bounded by $D\in\NN$; let $H$ consist of $l> 1$ connected components. Let $\epsilon<1$ be a proximity parameter and let $|V|=n$ and $|V_H|=m$. Let $k_{\text{max}}$ the maximum number of source components of one of the connected components of $H$ and let $k_{\text{min}}$ the minimum number of source components of one of them.

Then there is a property testing algorithm for $H$-freeness, which in Graphs with $n$ vertices and both vertex indegree and outdegree bounded by $D$ has a query complexity of $\Ovon\left( l\log l \cdot D^m\left(\frac{m}{\epsilon} \right)^{1/{k_\text{min}}}n^{1-1/{k_\text{max}}} \right)$.
\end{corollary}

\subsection{A Property Testing Algorithm for $3$-Star-Freeness}\label{kap_3stern}

We will now consider the problem of testing for a certain class of forbidden subgraphs, the class of all orientations of $3$-stars. The algorithm that we introduce distinguishes between directed graphs that contain many occurences of $3$-stars, regardless of how the edges of these occurences are oriented. We make two major additional assumptions: The first is that the input graph is weakly connected. The second assumption is that the input graph does not contain any double edges, i.e. pairs of edges in both directions between two vertices; thus, the undirected degree of a vertex is equal to its number of adjacent vertices. The latter assumption is only made for simplicity and can be dropped using a simple reduction, as we will see later. The former assumption is, however, necessary, since without this assumption testing $3$-star-freeness requires $\Omega(n^{2/3})$ queries. 

We will additionally assume that the undirected degree of each vertex of the input graph is at most $D\in\NN_{\geq 3}$. At last, we can assume $n=\omega(\frac{D^4}{\epsilon^3})$: Else an algorithm could read all the input graph's edges with at most $Dn = \Ovon(\frac{D^3}{\epsilon^{3/2}})$ queries and thus solve the problem deterministically.

In the following we will simply use the term $k$-star for any $k$-star orientation in directed graph contexts. We have to distinguish several types of $k$-stars, depending on their number of incoming and outgoing edges:

\begin{definition}
Let $G$ be a directed $k$-star and let $2\leq m\leq k$. We call $G$ \emph{incoming} $m$-star, if for at least $m$ of the edges of $G$ the central vertex of $G$ is the target vertex. We call $G$ \emph{outgoing} $m$-star, if for at least $m$ of the edges of $G$ the central vertex of $G$ is the source vertex. If $G$ is an incoming respectively outgoing $m$-star, we call the central vertex of $G$ outgoing respectively incoming $k$-star vertex.
\end{definition}

As noted above, a simple algorithm for testing $3$-star freeness works as follows: Run the algorithm {\sc TestSubgraphFreenessAmplified} from the previous section for each of the $4$  (neglecting isomorphism) possible orientations of $3$-stars and for a proximity parameter of $\epsilon/4$; if one of the calls to {\sc TestSubgraphFreenessAmplified} rejects, return false; else return true. This algorithm has a query complexity of $\Theta(n^{2/3})$, since an incoming $3$-star without outgoing edges has $3$ source components. The algorithm that is introduced in this chapter will instead have a query complexity of $\Ovon(n^{1/2})$, but will in exchange have a two-sided error.

Since at this point the input graphs are restricted to weakly connected directed graphs without double edges, we can formulate the problem differently: We want to distinguish directed graphs that are orientations of a single circle or a line from those that are $\epsilon$-far from doing so. This is not trivial, since due to the restriction that only outugoing edges of a vertex are possible to query, the portion of the graph explored by a breadth first search may be very small, particularly for vertices that have only incoming edges.

The algorithm makes use of a certain property of directed graphs that contain many occurences of incoming $3$-stars without outgoing edges: In such graphs, the number of edges that are part of incoming $2$-stars is significantly larger than the number of edges that are part of outgoing two-stars (note that every incoming $3$-star vertex is also an incoming $2$-star vertex). Moreover, the probability of sampling two incloming edges of a vertex grows quadratically in the total number of its incoming edges. Thus, in a graph that contains many incoming $3$-stars, the number of incoming $2$-stars found by collisions statistics on the target vertices of edges will ne disproportionately high in comparision to the number of incoming $2$-stars in this graph with high probability; on the other hand, in a graph is $3$-star-free, these values will be roughly the same with high probability, as we will see. Such collision statistics only require $\Ovon(n^{1/2})$ edge samples.

For testing for any $3$-stars that have incoming edges, we can rely on calling {\sc TestSubgraphFreeness} as sketched above. Since we now only have to call this algorithm for orientations with at most $2$ incoming edges, the query complexity for all these calls is $\Ovon(n^{1/2})$. Thus, we get a total query complexity of $\Ovon(n^{1/2})$. For formalizing these ideas, we state the algorithm {\sc Test3StarFreeness}.

\begin{Algorithm}[tb]
\noindent\centering
\shadowbox{
	\begin{minipage}{8cm}
		\begin{tabbing}
			~~~~\= ~~~~\= ~~~~\= ~~~~ \= ~~~~ \= \kill
			\textsc{\textbf{Test3StarFreeness}}($n,G,\epsilon$)\\
					\>\textbf{if} \textsc{EstimateEdgeCount}($n,G,\frac{\epsilon}{16}$)$>n+\frac{\epsilon n}{16}$ \textbf{then return} false\\
					\> \textbf{Foreach} $3$-star orientation $H$ that has at least one outgoing edge \\
					\> \> \textbf{if not} \textsc{TestSubgraphFreenessAmplified}($G,H,\frac{\epsilon}{192D},\frac{1}{6}$)\\
					\> \> \> \textbf{then return} false\\

					\> sample $s_{\ref{alg_3stern}.1}=\frac{48}{\epsilon}$ vertices $v_1, \ldots, v_{\ref{alg_3stern}.1}$ of $G$ u.i.d.\\
					\> $\hat k \leftarrow \frac{n}{s_{\ref{alg_3stern}.1}}$ times the number of outgoing $2$-star vertices in $\{v_1,\ldots, v_{\ref{alg_3stern}.1}\}$\\
					\> $\hat k\leftarrow \hat k + \frac{\epsilon n}{12}$\\

					\> sample each edge of $G$ with probability $p=\frac{128D}{\epsilon^{3/2}\sqrt{n}}$\\
					\> \textbf{if} more than $s_{\ref{alg_3stern}.2}=\frac{2048\cdot D\sqrt{n}}{\epsilon^{3/2}}$ are sampled \textbf{then return} false\\
					\> $\hat c\leftarrow\frac{1}{p^2}$ times the number of collisions on target vertices of sampled edges\\
					\> \textbf{if} $\hat c<\frac{3}{2}\epsilon n$ \textbf{then return} true\\
					\> \textbf{if} $\hat r:=\frac{\hat c}{\hat k}> 1+\frac{\epsilon}{24}$ \textbf{then return} false\\
					\>\textbf{else return} true
		\end{tabbing}
	\end{minipage}
}
\caption{{\sc Test3StarFreeness}}\label{alg_3stern}
\end{Algorithm}

We will at first introduce a couple of Lemmas that state structural properties of weakly connected directed graphs. All these Lemmas are connected to the notion of the \emph{balance} of a directed graph:

\begin{definition}
Let $G=(V,E)$ be a directed graph. The \emph{balance} $\calB(G)$ of $G$ is the difference between the number of outgoing $2$-star vertices and the number of incoming $2$-star vertices in $G$, i.e.:
$$\calB(G)=\left|~\#\{ v\in V| \degout(v)\geq 2 \} - \#\{ v\in V| \degin(v)\geq 2 \}  ~\right|$$
\end{definition}

We will show several upper bounds on the balance of directed graphs, depending on their structure. We will start with graphs that are orientations of an undirected circle graph.

\begin{figure}
	\centering
	\begin{tikzpicture}[thick]
		\node at (-1.2,1.9) {\Large{(a)}};
		\foreach \i/\x/\y in {1/0.98/-1.14, 3/-0.78/-1.28, 4/-1.39/-0.58, 6/-0.98/1.14, 9/1.39/0.58, 10/1.46/-0.35}
			\node (p\i) [draw=black,fill=gray70, minimum size = 0.15cm,inner sep = 0cm,circle] at (\x,\y) {};

		\foreach \i/\x/\y in {2/0.12/-1.50, 7/-0.12/1.50}
			\node (p\i) [draw=black,fill=red, minimum size = 0.15cm,inner sep = 0cm,circle] at (\x,\y) {};

		\foreach \i/\x/\y in {5/-1.46/0.35, 8/0.78/1.28}
			\node (p\i) [draw=black,fill=blue!60, minimum size = 0.15cm,inner sep = 0cm,circle] at (\x,\y) {};

		\foreach \i/\j in {1/2,3/2,4/3,5/4,5/6,6/7,8/7,8/9,9/10,10/1}
			\draw[->] (p\i) -- (p\j);

		\draw[dashed, color=lightgray] (2.2,-1.5) -- (2.2, 1.7);
	\end{tikzpicture}~~~
	\begin{tikzpicture}[thick]
		\node at (.5,1.3) {\Large{(b)}};
		\node at (-.8,.4) {\Large{$e$}};	

		\foreach \i/\x/\y in {2/1/0, 6/-3/1, 7/-2/-2, 8/0/-2, 11/3/-1, 12/3/0, 13/4/1}
			\node (p\i) [draw=black,fill=gray70, minimum size = 0.15cm,inner sep = 0cm,circle] at (\x,\y) {};

		\foreach \i/\x/\y in {5/-2/1, 4/-1/-1}
			\node (p\i) [draw=black,fill=red, minimum size = 0.15cm,inner sep = 0cm,circle] at (\x,\y) {};

		\foreach \i/\x/\y in {1/0/0, 3/-1/1, 10/3/1}
			\node (p\i) [draw=black,fill=blue!60, minimum size = 0.15cm,inner sep = 0cm,circle] at (\x,\y) {};

		\foreach \i/\x/\y in {9/2/0}
			\node (p\i) [draw=black,fill=violet!70, minimum size = 0.15cm,inner sep = 0cm,circle] at (\x,\y) {};

		\foreach \i/\j in {1/2, 1/4, 3/5, 6/5, 7/4, 4/8, 2/9, 10/9, 9/11, 9/12, 10/13}
			\draw[->] (p\i) -- (p\j);

		\draw[->, dashed] (p3) -- (p1);		
	\end{tikzpicture}
	\caption{Orientations of undirected graphs: (a) Circle graph; (b) Tree; a feasible edge $e$ for choosing in the inductive step of the proof of Lemma \ref{lemma_baumbalance} is drawn dashed. In both cases incoming $2$-star vertices are painted in red and outgoing $2$-star vertices are painted in blue.}\label{abb_balance}
\end{figure}
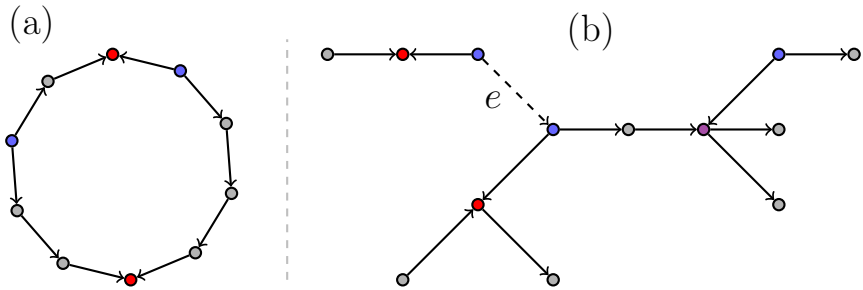

\begin{lemma}\label{lemma_kreisbalance}
Let $G=(V,E)$ be an orientation of an undirected circle graph. Then, $\calB(G)=0$.
\end{lemma}

\begin{proof}
Since each vertex in $G$ has an undirected degree of two, incoming $2$-star vertices do not have outgoing edges and outgoing $2$-star vertices do not have incoming edges. Each vertex that is neither of them has exactly one incoming edge and one outgoint edge. Furthermore, the total number of incoming edges in $G$ equals the total number of outgoing edges, and hence, by the above considerations, for each vertex that has two outgoing edges, there must be a vertex that has two incoming edges and vice versa. Thus the number of incoming $2$-star vertices in $G$ equals the number of outgoint $2$-star vertices, i.e. $\calB(G)=0$.
\end{proof}

The following proofs make use of the techniques of contracting and re-expanding and of deleting and re-inserting edges of graphs. Observe that inserting an edge into a graph can only change its balance by one: Only the target vertex of the edge can become a new incoming $2$-star, and only its source vertex can become a new outgoing $2$-star. Of course the same holds true in reverse direction for deleting an edge.

We continue by showing that orientations of line-shaped graphs have a balance of at most $1$, and, if the balance of such a graph is not $0$, both edges at the ends of the graph have the same orientation. A line-shaped graph is a tree that has only two leafs.

\begin{lemma}\label{lemma_linienbalance}
Let $G=(V,E)$ be an orientation of a line-shaped graph. Then, $\calB(G)\leq 1$. If $\calB(G)\neq 0$, one of the following statements holds:
\begin{itemize}
	\item The number of outgoing $2$-star vertices in $G$ exceeds the number of incoming $2$-star vertices by $1$ and both leafs of $G$ have incoming edges.
	\item The number of incoming $2$-star vertices in $G$ exceeds the number of outgoing $2$-star vertices by $1$ and both leafs of $G$ have outgoing edges.
\end{itemize}
\end{lemma}

\begin{proof}
At first we prove $\calB(G)\leq 1$. Consider an arbitratily oriented edge $e$ to be inserted between the two leafs of $G$. The resulting graph $G'$ is an orientation of a circle and, according to lemma \ref{lemma_kreisbalance}, has a balance of $0$. Since by re-deleting $e$ from $G'$ the balance can change by at most $1$, we have $\calB(G)\leq 1$.

Now assume that $\calB(G)=1$. Then the insertion of $e$ must have created a new $2$-star vertex of the type that appears less in $G$ and cannot have created one of the other type, since $G'$ has a balance of $0$. Thus, if $G$ has more incoming $2$-star vertices than outgoing $2$-star vertices, both leafs have an outgoing edge, so that the insertion of $e$ creates a new outgoing $2$-star vertex at one of them and does not create a new $2$-star vertex at the other. Equvialently, if $G$ has more outgoing $2$-star vertices than incoming $2$-star vertices, both leafs have an incoming edge.
\end{proof}

We will now show that the maximum balance of a tree depends on its number of leafs.

\begin{lemma}\label{lemma_baumbalance}
Let $G=(V,E)$ be an orientation of a tree with $k$ leafs. Then $\calB(G)\leq k-1$.
\end{lemma}

\begin{proof}
We prove the lemma by induction over the number $k$ of leafs of the tree. For the base case let $G$ be an orientation of a tree with $k=2$ leafs. Then, by Lemma \ref{lemma_linienbalance}, we have $\calB(G)\leq 1=k-1$.

For the induction step we assume that $G$ is an orientation of a tree with $k>2$ leafs and that every orientations of a tree with at most $m<k$ leafs has a balance of at most $m-1$.

Since $G$ has at least three leaf vertices, there exists at least one $3$-star vertex $v$ in $G$. If we delete an edge $e$ incident to $v$, $G$ decomposes into two subgraphs $G_1$ and $G_2$ which are not connected to each other. We can choose $v$ and $e$ in such a way that $G_2$ is a single vertex or an orientation of a line-shaped graph by starting at an arbitrary leaf of $G$ and traversing along the edges of $G$ until the first $3$-star vertex is reached. Let $v$ be a vertex found in such a way and let $e$ be the last edge visited when traversing from the corresponding leaf to $v$ (see figure \ref{abb_balance}).

$G_1$ has $k-1$ leafs and, by the induction hypothesis, is guaranteed to have a balance of at most $k-2$; $G_2$ has a balance of $0$, if it consists of a single vertex, and, by the induction hypothesis, a balance of at most $1$ elsewise.

We now re-insert $e$: We have $\calB(G)\leq \calB(G_1)+\calB(G_2)+1$. If $\calB(G_2)=0$ this immediately yields $\calB(G)\leq k-1$; if the majority of $2$-stars in $G_1$ is incoming and the majority of $2$-stars in $G_2$ is outgoing (or vice versa), $\calB(G)\leq k-1$ follows immediately, too.

Now assume that neither of it is the case, i.e. both $G_1$ and $G_2$ the same type of $2$-stars has the majority and the balance $G_2$ is $1$; without loss of generality, assume that the majority of $2$-stars in both graphs is incoming. Let $u$ be the vertex of $G_2$ that $e$ is connected to. Then, by Lemma \ref{lemma_linienbalance}, $u$ has an outgoing edge. Thus, if $e=(u,v)$, a new outgoing $2$-star at $u$ is created, and at most one new incoming $2$-star at $v$ is created. If $e=(v,u)$, no incoming $2$-star at $u$ is created, since the other edge incident to $u$ is incoming, and, since $e$ is an outgoing edge at $v$, no additional incoming $2$-star at $v$ is created. Thus, the number of incoming $2$-star created by re-inserting $e$ is bounded by the number of outgoing $2$-stars created, and thus we have $\calB(G)\leq\calB(G_1)+\calB(G_2)\leq k-2 + 1 = k-1$.
\end{proof}

Finally we will derive an upper bound on the balance of arbitrary weakly connected graphs. We will need the following observation for this:

\begin{observation}\label{obs_blaetterzahl}
Let $G=(V,E)$ be an orientation of a tree and let $C_3\subseteq V$ be the set of $3$-star vertices of $G$. Then $G$ has exactly $2-2|C_3|+\sum_{v\in C_3}\deg(v)$ leafs.
\end{observation}

\begin{proof}
We show the observation by induction over the number of $3$-star vertices in $G$. For the base case observe that if $|C_3|=0$, $G$ is an orientation of a line-shaped graph and thus has $2$ leafs; hence the observation holds. For $|C_3|=1$, assume that $v$ is the only $3$-star vertex in $G$ and $v$ has an undirected degree of $k$. Then, $G$ consists of orientations of $k$ lines, which are connected in $v$; $G$ thus has $k$ leafs and we have $2-2\cdot |C_3| + \sum_{v\in C_3} = k$; hence the observation holds if $G$ has at most one $3$-star vertex.

For the induction step we assume that $G$ is a graph with $|C_3|=i\geq 2$ $3$-star vertices and that the observation holds for all graphs $G'$ that have $j<i$ $3$-star vertices. Since $G$ is a tree and contains at least $2$ $3$-star vertices, there exists an edge $e$ which is incident to a $3$-star vertex $v$ and lies on a direct path from $v$ to another $3$-star vertex $u$, but all the other edges incident to $v$ lie on paths to leafs of $G$. We delete all the vertices and edges from $G$ that form paths from $v$ to leafs, but we sustain $v$ and $e$. By the choice of $v$, $v$ is a leaf in the resulting graphs $G'$, and the set of $3$-star vertices in $G'$ is $C_3'=C_3\backslash\{v\}$. Thus, $G'$ has less than less than $i$ $3$-stars and by the induction hypothesis has at most $2-2|C_3'|+\sum_{v\in C_3'}\deg(v)$ leafs; furthermore, the degree of every vertex $v\in C_3'$ in $G'$ equals its degree in $G$.

Now let $\tilde G$ be the subgraph of $G$ induced by the deleted vertices and $v$: Since by construction $\tilde G$ can have no $3$-star vertices other than $v$, $\tilde G$ contains at most a single $3$-star. Thus, by the base case, the number of leafs of $\tilde G$ equals the undirected degree of $v$ in $\tilde G$; let $k$ be this number. Now we join $G'$ and $\tilde G$: $v$ now has a degree of $k+1$, and in the resulting graph $G$, $v$ is not a leaf any more compared to $G'$, but $G$ contains $k$ additional leafs from $\tilde G$; thus, the number of leafs in $G$ is larger than that of $G'$ by $\deg(v)-2$ and hence equals

\begin{align*}
	\deg(v)-2 + 2-2|C_3'|+\sum_{v\in C_3'}\deg(v) &= 2-2|C_3'\cup\{v\}|+\sum_{v\in C_3'\cup\{v\}}\deg(v)\\
		 &= 2-2|C_3|+\sum_{v\in C_3}\deg(v).
\end{align*}
\end{proof}

Now we can proceed by proving an upper bound on the balance of an arbitrary weakly connected graph:

\begin{lemma}\label{lemma_balanceallgemein}
Let $G=(V,E)$ be a weakly connected graph that has $(1+\delta)|V|$ edges, $\delta>0$. Also let $C_3\subseteq V$ be the set of $3$-star vertices of $G$. Then $G$ has a balance of at most $2+\delta |V| - 2|C_3| + \sum_{v\in C_3}\deg(v)$.
\end{lemma}

\begin{proof}
Delete $\delta |V|+1$ edges from $G$ in such a way that $G$ remains weakly connected; Since the resulting graph $G'$ is weakly connected and has $V-1$ edges, $G'$ is an orientation of a tree. By  Lemma \ref{lemma_baumbalance} the balance of $G'$ is bounded by its number of leafs, which, by observation \ref{obs_blaetterzahl}, is $2-2|C_3| +\sum_{v\in C_3}\deg(v)$; hence it holds $\calB(G')\leq 1-2|C_3| +\sum_{v\in C_3}\deg(v)$.

By re-inserting all the deleted edges into $G'$, we get $G$. Since every re-inserted edge can only alter the balance by $1$, we have $\calB(G)\leq \calB(G_1) + \delta |V| + 1\leq 2 + \delta |V|- 2|C_3| +\sum_{v\in C_3}\deg(v)$.
\end{proof}

The following simple Lemma states that every weakly connected graph whose number of edges is larger than its number of vertices, contains at least one $3$-star vertices. Hence, a graph can be rejected if it contains too many edges.

\begin{lemma}\label{lemma_kanten_3sterne}
Let $G=(V,E)$ be a weakly connected directed graph with $n$ vertices and $m>n$ edges. Then, $G$ contains at least $1$ $3$-star vertex.
\end{lemma}

\begin{proof}
Since every edge of $G$ has two incident vertices, the sum of all undirected vertex degrees in $G$ is $2m$. Hence, the average undirected vertex degree in $G$ is $2m/n>2$; thus, there exists at least one vertex that has an undirected degree of more than $2$ and therefore is a $3$-star vertex (note that, by the assumption from the beginning of the section, $G$ does not contain any double-edges).
\end{proof}

We will use the fact stated by Lemma \ref{lemma_kanten_3sterne} to reject every input graph, whose edge count exceeds a certain number. To estimate the edge count of a given directed graph, we use the algorithm {\sc EstimateEdgeCount}. The algorithm samples $\Ovon(D/\epsilon)$ pairs $(v,i)\in V(G)\times \{1,\ldots, D\}$ of a vertex and the number of a slot of its adjacency list uniformly and independently distributed at random; it then checks for all those pairs whether the corresponding edge exists and extrapolates the fraction of edges found. Note that a similar technique has already been used in a property testing algorithm for acyclicity in bounded-degree graphs by Goldreich and Ron \cite{GR97}. 

\begin{Algorithm}[tb]
\noindent\centering
\shadowbox{
	\begin{minipage}{8cm}
		\begin{tabbing}
			~~~~\= ~~~~\= ~~~~\= ~~~~ \= \kill
			\textsc{\textbf{EstimateEdgeCount}}($n,G,\epsilon$)\\
					\> $\hat m \leftarrow 0$\\
					\> Sample $s_{\ref{alg_kantenzahl}}=\frac{2D}{\epsilon}$ pairs $(v_1,k_1), \ldots, (v_{s_{\ref{alg_kantenzahl}}}, k_{s_{\ref{alg_kantenzahl}}})$ of a vertex\\
					\>\>\>\> and a number of an adjacency list slot, i.e. $(v_i,k_i)\in V\times\{1\ldots D\}$\\									
	\> \textbf{for} $i\leftarrow 1$ \textbf{to} $s_{\ref{alg_kantenzahl}}$ \textbf{do}\\
					\>\> \textbf{if} $v_i$ has a $k_i$-th edge \textbf{then} $\hat m\leftarrow\hat m + \frac{Dn}{s_{\ref{alg_kantenzahl}}}$\\
					\> \textbf{return} $\hat m$
		\end{tabbing}
	\end{minipage}
}\\
\caption{{\sc EstimateEdgeCount}}\label{alg_kantenzahl}
\end{Algorithm}

\begin{lemma}\label{lemma_kantenzahl}
Let $G=(V,E)$ be a directed graph with $n$ vertices and $m$ edges and let the undirected degree of every vertex of $G$ be bounded by $D\in\NN$; let $\epsilon<1$ be a proximity parameter. Then {\sc EstimateEdgeCount} returns a value $\hat m$ that, with probability at least $1-2e^{-4}$, satisfies
$$m-\epsilon n \leq \hat m \leq m+\epsilon n.$$
The running time of {\sc EstimateEdgeCount} is $\Ovon(D/\epsilon)$.
\end{lemma}

\begin{proof}
Let $X_i$ be an indicator random variable for the event that $v_i$ has a $k_i$-th neighbor, $i=1,\ldots ,s_{\ref{alg_kantenzahl}}$; let $X=\sum_{1\leq i\leq s_{\ref{alg_kantenzahl}}} X_i$. Each of the $m$ edges of $G$ is stored in exactly one adjacency list slot as an outgoing edge and there are $Dn$ adjacency list slots in total. Since $(v_i, k_i)$ is chosen u.i.d. at random among all pairs of a vertex and an adjacency list slot, it holds
$$\E[X_i]=\Pr[X_i=1]=\frac{m}{Dn}$$
and, by the linearity of expectation,
\begin{align*}
\E[\hat m] & = \E\left[\frac{Dn}{s_{\ref{alg_kantenzahl}}}\cdot X\right] = \E\left[\frac{Dn}{s_{\ref{alg_kantenzahl}}}\cdot \sum_{1\leq i\leq s_{\ref{alg_kantenzahl}}} X_i\right] \\
	&=\frac{Dn}{s_{\ref{alg_kantenzahl}}}\cdot \sum_{1\leq i\leq s_{\ref{alg_kantenzahl}}} E[X_i]= \frac{Dn}{s_{\ref{alg_kantenzahl}}}\cdot \sum_{1\leq i\leq s_{\ref{alg_kantenzahl}}} \frac{m}{Dn} = m.
\end{align*}

It remains to show that the probability that $\hat m$ differs from $m$ by more than $\epsilon n$ is at most $2e^{-4}$. For this purpose, we use an additive Chernoff Bound:

$$\Pr\left[\left|\hat m - \E[\hat m]\right| > \epsilon n\right] = \Pr\left[ |X-E[X]| > \frac{\epsilon s_{\ref{alg_kantenzahl}}}{D}  \right]\leq 2e^{-\epsilon^2s_{\ref{alg_kantenzahl}}^2/D^2} = 2e^{-4}.$$

The running time of {\sc EstimateEdgeCount} directly follows from the number of samples taken.
\end{proof}

Analogously to the above proof one can show that the value $\hat k$ that is computed in algorithm {\sc Tes\-te\-3\-Stern\-Frei\-heit} is a good estimate of the number of outgoing $2$-stars in the input graph:

\begin{lemma}\label{lemma_aus2sternzahl}
Let $G=(V,E)$ be a directed graph with $n$ vertices and a bounded undirected vertex degree of $D\in\NN$. Let $\epsilon$ be a proximity parameter and assume that $G$ contains $k\leq n$ outgoing $2$-star vertices. Then, for the value $\hat k$ computed by {\sc Test3StarFreeness}($n,G,\epsilon$) it holds
$$k\leq \hat k \leq k+ \frac{\epsilon}{6}n$$
with a probability of at least $1-2e^{-16}$.
\end{lemma}

\begin{proof}
We consider the vertex sample drawn in the fifth line of the algorithm: Let $X_i$ be an indicator random variable for the event that the vertex $v_i$ is an outgoing $2$-star vertex. Both the probability for this event and the expected value of $X_i$ are $k/n$. Let $X=\sum_{1\leq i\leq s_{\ref{alg_3stern}.1}}X_i$; then, $\hat k=\frac{\epsilon}{12}n+\frac{n}{s_{\ref{alg_3stern}.1}}X$. By the linearity of expectation we get
$$\E\left[\hat k\right] =  \frac{\epsilon}{12}n+\frac{n}{s_{\ref{alg_3stern}.1}}\E[X] = \frac{\epsilon}{12}n+\frac{n}{s_{\ref{alg_3stern}.1}}\left( \sum_{1\leq i\leq s_{\ref{alg_3stern}.1}}\frac{k}{n} \right) = k + \frac{\epsilon}{12}n.$$
To bound the probability of $\hat k$ deviating from its expectation by more than $\frac{\epsilon}{12}n$ we again use an additive Chernoff Bound:
$$\Pr\left[\left| \hat k - \E[\hat k]\right| >\epsilon n/12 \right] = \Pr\left[\left| X - \E[X]\right| > \frac{\epsilon}{12} s_{\ref{alg_3stern}.1} \right]\leq 2e^{-\epsilon^2s_{\ref{alg_3stern}.1}^2/144}=2e^{-16}.$$
\end{proof}

We can now start with proving the correctness of {\sc Test3StarFreeness}. We divide the proof into four parts: The first part handles the case that the input graph $G$ (having $n$ vertices) is $3$-star free, the remaining three parts the case that $G$ ist $\epsilon$-far from $3$-star free. In this case we consider the following three subcases:
\begin{itemize}
	\item $G$ has more than $n+\frac{\epsilon}{8}n$ edges;
	\item $G$ has at most $n+\frac{\epsilon}{8}n$ edges and contains at least $\frac{\epsilon}{16D}n$ $3$-star vertices that have at least one outgoing edge;
	\item $G$ has at most $n+\frac{\epsilon}{8}n$ edges and contains less than $\frac{\epsilon}{16D}n$ $3$-star vertices that have at least one outgoing edge (and thus there are many $3$-star vertices that have only incoming edges).
\end{itemize} 

For every possible input one of the above cases holds. We start our analysis with the three subcases where we assume $G$ to be $\epsilon$-far from $3$-star free. In case $G$ has more than $n+\frac{\epsilon}{8}n$ edges, {\sc Test3StarFreeness} will reject in the first line with probability at least $2e^{-4}$; this follows directly from Lemma \ref{lemma_kantenzahl}:

\begin{lemma}
Let $G=(V,E)$ be a weakly connected directed graph with $n$ vertices and a bounded undirected vertex degree of $D\in\NN$. Let $\epsilon<1$ be a proximity parameter and assume that $G$ has $m> n+\frac{\epsilon}{8}n$ edges. Then, {\sc Test3StarFreeness}($n,G,\epsilon$) returns \emph{false} with a probability of at least $1-2e^{-4}$ \emph{false}.
\end{lemma}

The next case we consider is that $G$ has at most $n+\frac{\epsilon}{8}n$ edges and contains at least $\frac{\epsilon}{16D}n$ $3$-star vertices that have an outgoing edge. Each of these $3$-star vertices is the central vertex of an occurence of a $3$-star orientation in $G$, that has at least one incoming edge; there are (not considering isomorphism) $3$ such orientations, and hence one of them has at least $\frac{\epsilon}{48D}n$ occurences in $G$.

Since deleting an edge from $G$ can at most remove two of these occurences, at least $\frac{\epsilon}{96D}n\geq \frac{\epsilon}{192D}n + 1$ edge modifications in $G$ are necessary to remove all the occurences of this orientation\footnote{$n\geq \frac{192D}{\epsilon}$ can be assumed since we have assumed $n=\omega\left(\frac{D^4}{\epsilon^2}\right)$ above.}, and thus $G$ is $\frac{\epsilon}{192D}$-far from freeness of this $3$-star orientation. Hence and by Lemma \ref{lemma_subgraphtestverstärkt}, the corresponding call to {\sc TestSubgraphFreenessAmplified} in the third line of {\sc Test3StarFreeness} returns \emph{false} with a probability of at least $5/6$ \emph{false}. These considerations yield the following Lemma:

\begin{lemma}\label{lemma_ausg3sterne}
Let $G=(V,E)$ be a directed graph with $n$ vertices and a bounded undirected vertex degree of $D\in\NN$. Let $\epsilon<1$ be a proximity parameter and assume that $G$ has $m\leq n+\frac{\epsilon}{8}n$ edges and contains at least $\frac{\epsilon}{16D}n$ $3$-star vertices that have at least one outgoing edge each. Then, {\sc Test3StarFreeness}($n,G,\epsilon$) returns \emph{false} with a probabilty of at least $5/6$.
\end{lemma}

For the case of the input graph $G$ being $\epsilon$-far from $3$-star free, it remains to show that {\sc Test3StarFreeness} rejects with high probability if $G$ contains less than $\frac{\epsilon}{16D}n$ $3$-stars that have outgoing edges.

\begin{figure}
	\centering
	\begin{tikzpicture}[thick]
		\node at (0,.4) {\Large{$v$}};	
		\node at (-.8,.4) {\Large{$e_1$}};	
		\node at (.25,-.65) {\Large{$e_2$}};	

		\foreach \i/\x/\y in {6/-3/1, 7/-2/-2, 8/0/-2, 11/3/-2, 12/3/-1, 13/4/0, 14/1/1}
			\node (p\i) [draw=black,fill=gray70, minimum size = 0.15cm,inner sep = 0cm,circle] at (\x,\y) {};

		\foreach \i/\x/\y in {5/-2/1, 4/-1/-1, 1/0/0}
			\node (p\i) [draw=black,fill=red, minimum size = 0.15cm,inner sep = 0cm,circle] at (\x,\y) {};

		\foreach \i/\x/\y in {2/1/-1, 3/-1/1, 10/3/0}
			\node (p\i) [draw=black,fill=blue!60, minimum size = 0.15cm,inner sep = 0cm,circle] at (\x,\y) {};

		\foreach \i/\x/\y in {9/2/-1}
			\node (p\i) [draw=black,fill=violet!70, minimum size = 0.15cm,inner sep = 0cm,circle] at (\x,\y) {};

		\foreach \i/\j in {1/4, 3/5, 6/5, 7/4, 4/8, 2/9, 10/9, 9/11, 9/12, 10/13, 14/1}
			\draw[->] (p\i) -- (p\j);

		\draw[->, dashed] (p3) -- (p1);		
		\draw[->, dashed] (p2) -- (p1);	
	\end{tikzpicture}

	\caption{Two-way collision at an incoming $3$-star vertex $v$; the sampled edges $e_1$ and $e_2$ are highlighted and the vertex colouring follows that of figure \ref{abb_balance}.}
\end{figure}
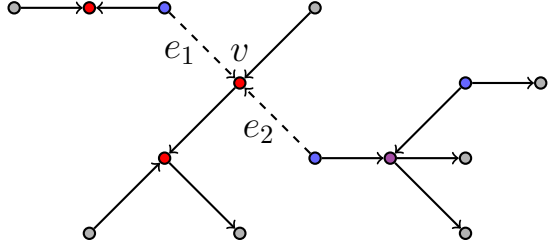

\begin{lemma}\label{lemma_eing3sterne}
Let $G=(V,E)$ be a weakly connected directed graph with $n$ vertices and a bounded undirected vertex degree of $D\in\NN$. Let $\epsilon<1$ be a proximity parameter and assume that $G$ is $\epsilon$-far from $3$-star free, has $m\leq n+\frac{\epsilon}{8}n$ edges and contains less than $\frac{\epsilon}{16D}n$ $3$-star vertices that have at least one outgoing edge each. Then, {\sc Test3StarFreeness}($n,G,\epsilon$) returns \emph{false} with a probabilty of at least $5/6$.
\end{lemma}

\begin{proof}
We will show that under the given assumptions $\hat r> 1+\frac{\epsilon}{24}$ holds in the first but las line of {\sc Test3StarFreeness} with a probability of at least $1-e^{-6}$; hence the algorithm rejects the input with high probability. The basic idea of this proof is the following: By Lemma \ref{lemma_balanceallgemein}, the maximum Balance of $G$ has a linear dependence on the number of its edges that are incident to $3$-star vertices; the probability of sampling at least two of the incoming edges of a $3$-star vertex, however, grows quadratically in the number of incoming edges that this vertex has. Since there are (as we will show) at least $\epsilon n$ $3$-star vertices in $G$ and almost all of them do not have outgoing edges, there are nearly $\epsilon n$ $3$-star vertices that have at least $3$ incoming edges each. Hence the probability of sampling a collision at one of them is relatively large: Counting the number of collisions and grossing up overestimates the number of outgoing $2$-star vertices (to which the outgoing $3$-star vertices belong) by more than the outgoing $2$-star vertices can outnumber the incoming $2$-star vertices due to the balance of $G$. Thus, the expected value of $\hat c$ will be much larger than that of $\hat k$ and the expected value of $\hat c$ will be much larger than $1$ hence.

We will need some additional definitions to concretize this idea. At fist, let $C_2$ be the set of incoming $2$-stars vertices that have exactly $2$ incoming edges and let $C_3$ be the set of incoming $3$-stars in $G$. We denote by $C_3'\subseteq C_3$ the set of incoming $3$-stars that do not have an outgoing edge and by $\tilde C_3\supseteq C_3$ the set of all $3$-star vertices in $G$. For the cardinalities of these sets we define  $c_2:=|C_2|$, $c_3:=|C_3|$, $c_3':=|C_3'|$ and $\tilde c_3:=|\tilde C_3|$. Then it holds $C:=C_2+C_3$ for the set of all incoming $2$-star vertices $C$ and $c:=c_2+c_3$ for its cardinality. Analogously to the last equation, we also say $\hat c = \hat c_2 + \hat c_3$, where $\hat c_2$ denotes the contribution of collisions on vertices in $C_2$ and $\hat c_3$ that of collisions on vertices in $C_3$.

$G$ contains at least $\epsilon n$ $3$-star vertices: If the number of $3$-star vertices was smaller, there would be less than $\epsilon Dn$ edges adjacent to them, and deleting all of these edges would remove every $3$-star vertex from $G$; this would be a contradiction to the assumption that $G$ is $\epsilon$-far from $3$-star free.

Since by the assumption of the lemma less than $\frac{\epsilon}{16D}n$ of these $3$-star vertices have incoming edges and $D\geq 3$, there are at least $c_3'>\left(1-\frac{1}{16D}\right)\epsilon n\geq \frac{47}{48}\epsilon n\label{eq_c3strichnachn}$ $3$-star vertices without outgoing edges. Since $C_3'$ is a subset of  $C_3$, it also holds \begin{equation}c_3\geq c_3'> \left(1-\frac{1}{16D}\right)\epsilon n\geq \frac{47}{48}\epsilon n.\label{eq_c3nachn}\end{equation}

We define $$l:=\sum_{v\in C_3}(\deg(v)-2),$$ and in the following we will derive an upper bound for $\hat k$ and a lower bound for $\hat c$; both bounds will relate on $l$, which will help us bounding $\hat r = \frac{\hat c}{\hat k}$ from below. Furthermore, it holds \begin{equation}l\geq c_3,\label{eq_c3nachl},\end{equation} since every vertex in $C_3$ has an undirected degree of at least $3$.

We start by proving an upper bound for $\hat k$. Since $G$ has at most $n+\frac{\epsilon}{8}n$ edges by the assumption of the Lemma, the balance of $G$ is at most
\begin{align*}
	\calB(G) & \leq \frac{\epsilon}{8}n + 2 + \sum_{v\in\tilde C_3}(\deg(v)-2)\\
			& = \frac{\epsilon}{8}n + 2 + \sum_{v\in C_3}(\deg(v)-2) + \sum_{v\in\tilde C_3\backslash C_3}(\deg(v)-2) \\
			& < \frac{\epsilon}{8}n + l + \frac{\epsilon}{16}n < l + \frac{\epsilon}{4}n,
\end{align*}
by Lemma \ref{lemma_balanceallgemein} and because $\tilde C_3 \backslash C_3$ only contains $3$-star vertices that have at least one outgoing edge; there are at most $\frac{\epsilon}{16D}n-1$ such $3$-star vertices by the assumption of the Lemma, and thus there are at most $D\left(\frac{\epsilon}{16D}n-1\right)< \frac{\epsilon}{16}n-2$ edges that are incident to one of them. 

Since $c$ is the number of incoming $2$-star vertices and the balance of $G$ is at most $l+\frac{\epsilon}{4}n$, the number of outgoing $2$-star vertices in $G$ is ar most $c+l+\frac{\epsilon}{4}n$. By Lemma \ref{lemma_aus2sternzahl} we get that, with probability at least $1-2e^{-16}$, it holds
$$\hat k\leq c+l+\frac{\epsilon}{4}n+\frac{\epsilon}{6}n< c_2 + c_3 + l + \frac{1}{2}c_3 \leq c_2 + \frac{5}{2}l;$$
the second inequality follows from the inequality $\frac{5}{12}\epsilon n< \frac{5}{12}\cdot\frac{48}{47}c_3< \frac{1}{2}c_3$, that we get from \eqref{eq_c3nachn}.

We will now bound $\hat c$ from above. Let $X_v$ be an indicator random variable for the event that for a vertex $v$ at least $2$ incoming edges are contained in the edge sample of {\sc Test3StarFreeness}, i.e., we have a ($2$-way) collision at $v$. Let $X^{(2)}=\sum_{v\in C_2}X_v$ und $X^{(3)}=\sum_{v\in C_3}X_v$. It holds
$$\E[X_v]=\Pr[X_v=1]=p^2(1-p)^{\indeg(v)-2}{\indeg(v) \choose 2}\geq \alpha p^2{\indeg(v) \choose 2}$$
for $\alpha:=\frac{47}{48}$ due to the minimum number of vertices we have assumed for $G$ and by the Bernoulli Inequality: $(1-p)^{\indeg(v)-2}\geq (1-p)^D\geq 1-pD =1- \frac{D^2}{\epsilon^{3/2}\sqrt{n}} \geq \frac{47}{48}$ for $n\geq \frac{2304D^4}{\epsilon^3}$.

Hence, for the contribution of collisions on incoming $2$-star vertices to $\hat c$,
\begin{align*}
	\E[\hat c_3] & = \frac{1}{p^2}\E\left[X^{(3)}\right] = \frac{1}{p^2}\sum_{v\in C_3}\E[X_v] \geq \alpha\sum_{v\in C_3}{\indeg(v) \choose 2}\geq \alpha\sum_{v\in C_3'}{\deg(v) \choose 2} \\
			& = \frac{\alpha}{2}\sum_{v\in C_3'}\deg(v)^2-\deg(v) \geq 3\alpha\sum_{v\in C_3'}\deg(v)-2 
\end{align*}
due to the linearity of expectation and since for every $3$-star vertex it holds $\deg(v)^2-\deg(v)\geq 6\cdot (\deg(v)-2)$. We want to relate the above sum to $l$, but in this sum only the vertices in $C_3'$ are considered. However, by the assumption of the Lemma there are at most $\frac{1}{16D}\epsilon n$ vertices in $C_3\backslash C_3'$ -- all of these vertices are $3$-star vertices that have at least one outgoing edge. The sum over the degrees of these vertices is at most $\frac{1}{16}\epsilon n\leq \frac{3}{47}l$ since it holds \begin{equation}l\geq c_3\geq \frac{47}{48}\epsilon n\label{eq_lnachn}\end{equation} by the inequalities \eqref{eq_c3nachl} and \eqref{eq_c3nachn}. Hence we can conclude

$$\sum_{v\in C_3'}\deg(v)-2 > \sum_{v\in C_3}\deg(v) - 2 ~~ - \sum_{v\in C_3\backslash C_3'}\deg(v) \geq l - \frac{3}{47}l = \frac{44}{47}l.$$

We can use this in the above bound for $\E[\hat c_3]$ and get $\E[\hat c_3] > 3\alpha l = \frac{11}{4}l\geq \frac{11}{4}\cdot\frac{47}{48}\epsilon n$ since $\alpha=\frac{47}{48}$ and due to inequality \eqref{eq_lnachn}.
Now we can apply a multiplicative Chernoff Bound, which yields an upper bound for the probability that $\hat c_3$ is smaller than this value by a factor of more than $(1-\epsilon/44)$:
\begin{align*}
	\Pr\left[\hat c_3<(1-\epsilon/44)\E[\hat c_3]\right] & = \Pr\left[X^{(3)}<(1-\epsilon/44)p^2\E\left[\hat c_3\right]\right]\leq \exp\left(-\frac{\epsilon^2p^2}{2\cdot 44^2}\E[\hat c_3]\right)\\&
	\leq \exp\left(-\frac{128^2 \epsilon^2D^2}{2\cdot 44^2\epsilon^3 n}\cdot\frac{11}{4}\cdot\frac{47}{48}\epsilon n\right) \leq \exp\left(-\frac{8\cdot 47 D^2}{33} \right)< e^{-100}.
\end{align*}
Analogously we can conclude $\Pr\left[\hat c_3<(1-\epsilon/24)\E[\hat c_3]\right]<e^{-600}$. In both cases we use the fact $D\geq 3$. 

For estimating $\hat c$ it remains to calculate the expected number of collisions on vertices in $C_2$. Since these vertices have exactly $2$ incoming edges, it holds
$$\E[\hat c_2] = \frac{1}{p^2}\sum_{v\in C_2}\E[X_v] = \frac{1}{p^2}\sum_{v\in C_2} p^2 = c_2.$$

Now we distinguish between two cases, $c_2\geq \frac{1}{16}l$ and $c_2< \frac{1}{16}l$. In the former case, $c_2$ is relatively large, and we can use a Chernoff Bound to show that $\hat c_2$ is a good estimate for $c_2$; then we can argue that the vertices in $C_3$ form at least roughly an $\epsilon$-fraction of the vertices in $C$ and thus $\hat c = \hat c_2 + \hat c_3$ is significantly larger than $c$. 

In the latter case, i.e., if $c_2$ is small, we will simply neglect the contribution of the vertices in $C_2$ to $\hat c$.

We start with the first case, i.e., we have $c_2\geq \frac{1}{16}l\geq \frac{47}{16\cdot 48}\epsilon n$ by inequality \eqref{eq_lnachn}, and by applying a multiplicative Chernoff Bound we conclude
\begin{align*}
	\Pr\left[\hat c_2<(1-\epsilon/24)\E[\hat c_2]\right] & = \Pr\left[X^{(2)}<(1-\epsilon/24)\E\left[X^{(2)}\right]\right]\leq \exp\left(-\frac{\epsilon^2p^2}{2\cdot 24^2}c_2\right)\\&
	\leq \exp\left(-\frac{128^2 \epsilon^2D^2}{2\cdot 24^2\epsilon^3 n}\cdot\frac{47\epsilon n}{16\cdot 48}\right) < e^{-7}.
\end{align*}

By adding the lower bounds for $\hat c_3$ and $\hat c_2$, we gain a lower bound for $\hat c$, i.e., with high probability it holds
\begin{align*}
	\hat r &= \frac{\hat c}{\hat k} = \frac{\hat c_2 + \hat c_3}{\hat k} \geq \frac{\left(1-\frac{\epsilon}{24}\right)\left(c_2+\frac{11}{4}c_3\right)}{c_2 + \frac{5}{2}c_3} = \left( 1 - \frac{\epsilon}{24} \right) \left( 1 + \frac{\frac{1}{4}c_3}{c_2+\frac{5}{2}c_3}  \right)\\
	& = \left( 1 - \frac{\epsilon}{24} \right) \left( 1 + \frac{1}{10}\left(\frac{\frac{5}{2}c_3}{c_2+\frac{5}{2}c_3} \right) \right).
\end{align*}
Furthermore we have $c_2\leq n-\left( 1-\frac{1}{16D}\right)\epsilon n$, since $G$ contains at least $\left( 1-\frac{1}{16D}\right)\epsilon n$ $3$-star vertices that do not have an incoming edge, and $c_3\geq \left( 1-\frac{1}{16D}\right)\epsilon n$ by inequality \eqref{eq_c3nachn} and since every such vertex is contained in $C_3$. We are allowed to insert this lower bound for $c_3$ into the fraction $\frac{5c_3/2}{c_2+5c_3/2}$ since the fraction is smaller than $1$ and by inserting both nominator and denominator are decreased by the same value. Hence we can conclude
\begin{align*}
	\hat r &\geq \left( 1 - \frac{\epsilon}{24} \right) \left( 1 + \frac{1}{10}\left(\frac{\frac{5}{2}\left( 1-\frac{1}{16D}\right)\epsilon n}{n - \left( 1-\frac{1}{16D}\right)\epsilon n + \frac{5}{2}\left( 1-\frac{1}{16D}\right)\epsilon n} \right) \right) \\
	& = \left( 1 - \frac{\epsilon}{24} \right) \left( 1 + \frac{\frac{1}{4}\epsilon}{\frac{1}{1-\frac{1}{16D}} + \frac{3}{2}\epsilon} \right).
\end{align*}
Since $1-\frac{1}{16D}\geq \frac{47}{48}$ and $\epsilon<1$ -- from which in particular follows $\epsilon>\epsilon^2$ --, we finally get
\begin{align*}
	\hat r &\geq  \left( 1 - \frac{\epsilon}{24} \right) \left( 1 + \frac{\frac{1}{4}\epsilon}{\frac{48}{47} + \frac{3}{2}\epsilon} \right) 
	\geq \left( 1 - \frac{\epsilon}{24} \right) \left( 1 + \frac{\epsilon}{\frac{192}{47} + 6} \right) = \left( 1 - \frac{\epsilon}{24} \right) \left( 1 + \frac{47\epsilon}{474} \right)\\
	& = 1 + \frac{(47\cdot 24 - 474)\epsilon - 47\epsilon^2}{24\cdot 474}\geq 1+\frac{607}{24\cdot 474}\epsilon > 1+\frac{1}{24}\epsilon.
\end{align*}
The probability that one of the bounds for $\hat c_2$ and $\hat c_3 $ that we have assumed here does not hold hold is, by the Union Bound, at most $e^{-600}+e^{-7}< e^{-6}$.

We will now consider the case that there are only very few incoming $2$-stars that are not incoming $3$-stars, i.e., $c_2<\frac{1}{16}l$. Since, by inequality \eqref{eq_c3nachl}, it holds $c_3\leq l$, we can conclude
\begin{align*}
	\hat r &= \frac{\hat c}{\hat k} \geq \frac{\hat c_3}{\hat k} \geq \frac{\left(1-\frac{\epsilon}{44}\right)\cdot \frac{11}{4}l}{c_2+\frac{5}{2}c_3} = \left(1-\frac{\epsilon}{44}\right)\cdot \frac{\frac{11}{4}l}{\frac{1}{16}l+\frac{5}{2}l} = \left(1-\frac{\epsilon}{44}\right)\cdot \frac{44}{41}\\
	&= \frac{44-\epsilon}{41} \geq \frac{41+2\epsilon}{41} > 1 + \frac{1}{24}\epsilon,
\end{align*}
with a probability of at least $e^{-100}<e^{-6}$.

We have by now proven that $\hat r>1+\frac{1}{24}\epsilon$ holds with a probability of $1-e^{-6}$ if the assumptions of the lemma are fulfilled and the algorithm {\sc Test3StarFreeness} gets to the first but last line; hence, the algorithm rejects the input with at least the above probability in this case. It remains to show that {\sc Test3StarFreeness} does not return true in the second but last line -- this would happen in case $\hat c\leq \frac{3}{2}\epsilon n$. We have, however, assumed that $\hat c_3\geq \left(1-\frac{\epsilon}{24}\right)\frac{11}{4}l$ holds, and by inequality \eqref{eq_lnachn} we can conclude $$\hat c\geq \hat c_3 \geq \left(1-\frac{\epsilon}{24}\right)\frac{11}{4}l \geq \frac{47}{48}\cdot\frac{11}{4}\epsilon n\geq 2\epsilon n$$ from this. Thus {\sc Test3StarFreeness} does not return true in the second but last line if $\hat r$ is estimated correctly.

Hence the probability that {\sc Test3StarFreeness} correctly returns \emph{false} is at least $1-e^{-6}>\frac{5}{6}$.
\end{proof}

It remains to show that {\sc Test3StarFreeness} works correctly if the input graph is $3$-star free.

\begin{lemma}\label{lemma_keine3sterne}
Let $G=(V,E)$ be a weakly connected directed $3$-star free graph with $n$ vertices and a bounded undirected vertex degree of $D\in\NN$. Let $\epsilon<1$ be a proximity parameter. Then, {\sc Test3StarFreeness}($n,G,\epsilon$) returns \emph{true} with a probabilty of at least $3/4$.
\end{lemma}

\begin{proof}
There are three possibilities for {\sc Test3StarFreeness} inadvertantly rejecting the input: By considerably overestimating the edge count in the first line, by drawing an edge sample of a size of more than $\frac{2048D\sqrt{n}}{\epsilon^{3/2}}$ and by computing a value $\hat r>1+\frac{\epsilon}{24}$. Since {\sc TestSubgraphFreenessAmplified} has an one-sided error and $G$ is $3$-star free, the subgraph freeness tests in the third line will always return true.

Since $G$ is weakly connected and does not contain any $3$-star vertex, $G$ contains at most $n$ edges (note that we still assume that there are no double edges in $G$). By Lemma \ref{lemma_kantenzahl} the probability that {\sc EstimateEdgeCount}($n,G,\epsilon/16$) returns a value larger than $n+\frac{\epsilon}{16}n$ is at most $2e^{-4}$; hence {\sc Test3StarFreeness} will not return false in the first line with a probability of at least $1-2e^{-4}$.

The expected number of edges that are drawn in the eigth line of {\sc Test3StarFreeness} is $pn=\frac{128D\sqrt{n}}{\epsilon^{3/2}}$; thus, by Markov's Inequality, the number of edges drawn will not be larger than $\frac{2048D\sqrt{n}}{\epsilon^{3/2}}=16pn$ with a probability of at least $\frac{15}{16}$. Hence {\sc Test3StarFreeness} will not return false in the eigth line with at least this probability.

It remains to show that the estimator $\hat e$ does not exceed $1+\frac{\epsilon}{24}$ with high probability. Let $k$ be the number of outgoing $2$-star vertices in $G$ and let $c$ be the number of incoming $2$-star vertices. By Lemma \ref{lemma_aus2sternzahl} it holds $\hat k\geq k$ with a probability of at least $1-2e^{-16}$, and, since $G$ is an orientation of a circle or a line-shaped graph, the balance of $G$ is at most $1$ by lemma \ref{lemma_kreisbalance} respectively by Lemma \ref{lemma_baumbalance}. Thus, we have $\hat k\geq c-1$ with a probability of at least $1-2e^{-16}$.

Now let $C_2$ be the set of incoming $2$-star vertices; let, for $v\in C_2$, $X_v$ be an indicator random variable for the event that both incoming edges of $v$ are in the edge sample of {\sc Test3StarFreeness}. Let $X:=\sum_{v\in C_2}X_v$. Es gilt $\E[X_v]=\Pr[X_v=1]=p^2$. For the expected value of $\hat c$ we conclude
$$\E[\hat c] = \frac{1}{p^2}X = \frac{1}{p^2}\sum_{v\in C_2}\E[X_v]=|C_2| = c. $$
Now at first assume $c\geq \frac{\epsilon}{4} n$. Then, by a multiplicative Chernoff Bound
\begin{align*}
	\Pr\left[\hat c\geq \left(1+\frac{\epsilon}{32}\right)c\right]&=\Pr\left[X\geq \left(1+\frac{\epsilon}{32}\right)E[X]\right]\leq \exp\left(-\frac{\epsilon^2 p^2c}{3\cdot 1024}\right)\\
	& = \exp\left(-\frac{128^2\epsilon^3D^2n}{3\cdot 4\cdot 1024\epsilon^3n}\right) = e^{-2D^2/3}\leq e^{-6}.
\end{align*}
Thus, by the Union Bound, with at least a probability of $1-2e^{-16}-e^{-6}>\frac{5}{6}$:
$$\hat r\leq\frac{\left(1+\frac{\epsilon}{32}\right)c}{c-1} = \left(1+\frac{\epsilon}{24}\right) + \frac{\left(1+\frac{\epsilon}{24}\right) - \frac{\epsilon}{96}c}{c-1} \leq \left(1+\frac{\epsilon}{24}\right) + \frac{\frac{25}{24} - \frac{\epsilon}{96\cdot 4}\epsilon n}{c-1},$$
and the last fraction is smaller or equal zero if $n\geq\frac{400}{\epsilon^2}$; however, this holds because of the minimum number of vertices of $G$ we have assumed above. Hence it holds $\hat r\leq 1+\frac{\epsilon}{24}$ with a probability of at least $\frac{5}{6}$.

Finally consider the second case for $c$, $c<\frac{\epsilon}{4}$. By Markov's Inequality we have $\Pr\left[\hat c>\frac{3}{2}\epsilon n\right]=\Pr\left[\hat c>6\E[c]\right]\leq \frac{1}{6}$, and hence {\sc Test3StarFreeness} returns true in the second but last line with a probability of at least $\frac{5}{6}$.

Thus, the probability that {\sc Test3StarFreeness}($n,G,\epsilon$) inadvertantly returns \emph{false} is at most $2e^{-4}+\frac{1}{16}+\frac{1}{6}<\frac{1}{4}$ by the Union-Bound.
\end{proof}

The correctness of {\sc Teste\-3\-Stern\-Frei\-heit} directly follows from Lemmas \ref{lemma_ausg3sterne}, \ref{lemma_eing3sterne} and \ref{lemma_keine3sterne}. The size of the edges sample drawn dominates the query complexity of {\sc Teste\-3\-Stern\-Frei\-heit}, which is, hence, $\Ovon\left(\frac{\sqrt{n}}{\epsilon^{3/2}}\right)$.

\begin{mtheorem}
Let $G=(V,E)$ be a weakly connected directed graph with $n$ vertices and a bounded undirected vertex degree of $D\in\NN$. Let $\epsilon<1$ be a proximity parameter and assume that $G$ does not have any double edges. Then, {\sc Test3StarFreeness}($n,G,\epsilon$) returns \emph{true} with a probabilty of at least $3/4$ if $G$ is $3$-star free, and returns \emph{false} with a probability of at least $\frac{5}{6}$ if $G$ is $\epsilon$-far from $3$-star free. The query complexity of the algorithm is $\Ovon\left(\frac{\sqrt{n}}{\epsilon^{3/2}}\right)$.
\end{mtheorem}

Finally we will reconsider the restriction that the input graphs are not allowed to have double edges. Assume that a graph $G$ has double edges: Then let $G'$ be the graph that results from deleting from every double edge in $G$ the edge whose starting vertex has the larger vertex number. Since in this way we delete exactly one edge from every double edge, $G'$ does not have any double edges; on the other hand, the number of adjacent vertices did not change for any vertex in $G'$. Hence, $G'$ contains exactly the same number of $3$-star vertices as $G$ and thus is $3$-star free if $G$ is $3$-star free and is $\epsilon$-far from $3$-star free if $G$ is $\epsilon$-far from $3$-star free. This reduction can be computed locally: For any sampled edge $(u,v)$ it has to be checked whether there is an edge $(v,u)$; if this is the case and $u$ has a larger vertex number than $v$, $(u,v)$ is considered non-existent and the sampled edge is dismissed and, instead, a new edge is drawn uniformly by random.

We can now run {\sc Test3StarFreeness} on $G'$ while construction $G'$ from $G$ locally. The expected number of edges we have to draw until the sample does not get dismissed is at most $2$, and by Markov's Inequality the number of edges sampled in $G$ is with high probability larger than the number of those returned to {\sc Test3StarFreeness} only by a constant number. If not so, the algorithm can abort computation by returning \emph{true}, which will result in a small increase of the error probability in case the input graph is $\epsilon$-far from $3$-star free:

\begin{corollary}
Let $G=(V,E)$ be a weakly connected directed graph with $n$ vertices and a bounded undirected vertex degree of $D\in\NN$. Let $\epsilon<1$ be a proximity parameter. Then, {\sc Test3StarFreeness}($n,G,\epsilon$) returns \emph{true} with a probabilty of at least $3/4$ if $G$ is $3$-star free, and returns \emph{false} with a probability of at least $\frac{3}{4}$ if $G$ is $\epsilon$-far from $3$-star free. The query complexity of the algorithm is $\Ovon\left(\frac{\sqrt{n}}{\epsilon^{3/2}}\right)$.
\end{corollary}

\section{Lower Bounds for Testing $3$-Star Freeness}\label{kap_3starlowerbound}

In this section we give two lower bounds on the query complexity of property-testing algorithms for $3$-star freeness: A lower bound of $\Omega(\sqrt{n})$ queries for testing $3$-star freeness in weakly connected graphs, and a lower bound of $\Omega(n^{2/3})$ queries if weak connectivity of the input graphs is not required. This means that the algorithm given in the last section is asymptotically optimal in the number $n$ of vertices of the input graph; moreover, it means that testing $3$-star freeness is easier if the input graphs are guaranteed to be weakly connected. This brings up the question whether there are more problems for which this holds true.

The first lower bound, that of $\Omega(\sqrt{n})$ for testing in weakly connected graphs, follows from graph classes given by Bender and Ron for their lower bound for testing strong connectiviy \cite{BR02}: Their class of strongly connected graphs consists of orientations of circles, where all edges have the same orientation. Such graphs are $3$-star free. Their class of graphs that are $\epsilon$-far from strongly connected consists of orientations of circles where all edges have the same orientation, but additionally there are more than $\epsilon Dn$ outer vertices that have an edge towards a circle vertex. Each circle vertex at most one incoming edge from an outer vertex, and since a circle vertex that has such an edge possesses three neighbours, it is a $3$-star vertex. By construction there are more than $\epsilon Dn$ such vertices, and since converting such a graph into a $3$-star free graph requires deletion of one incedent edge for each of these vertices, those graphs are $\epsilon$-far from $3$-star free. Since Bender and Ron show that the two classes of graphs cannot be distinguished with $o(\sqrt{n})$ queries, we can conclude the following corollary:

\begin{corollary}\label{kor_3starlower1}
In the adjacency list model for directed graphs where algorithms cannot query the incoming edges of a vertex, every property testing algorithm for $3$-star freeness has a query complexity of $\Omega(\sqrt{n})$, where $n$ is the number of vertices of the input graph. This even holds if the possible input graphs are restricted to be weakly connected.
\end{corollary}

In the remainder of this section we will derive a lower bound of $\Omega(n^{2/3})$ for the query complexity of every property testing algorithm for $3$-star freeness that does not require the input graph to be weakly connected. Note that for the above bound of $\Omega(n^{1/2})$ both graph classes considererd consist of weakly connected graphs. We will make use of a technique invented by Raskhodnikova et al. \cite{RRSS09} and construct two classes of problem instances -- one consisting of $3$-star free graphs, the other of $\epsilon$-far ones -- such that the distributions of incoming vertex degrees in graphs of these two classes have $2$ proportional moments. By the results of \cite{RRSS09} that means that every algorithm that measures a significant difference between these distributions uses $\Omega(n^{2/3})$.

The difficulty here is to show that measuring vertex in-degrees is the only thing an algorithm can do to compare graphs from these classes. Since this is hard to argue even for graphs that only consist of isolated $k$-stars, we will define a helper problem that is very similar to the distinct elements problem examined in \cite{RRSS09}. We will then reduce testing this helper problem to testing $3$-star freeness.

The helper problem is defined as follows: We are given a sequence $A$ of $m$ integers $A_i\in\NN$, $i=1,\ldots m$; each value occurs in at most $3$ elements of $A$. Let $l$ be the number of distinct values that occur in $A$: We assume that all these values are from $\{1,\ldots, l\}$. If there is no value $a\in\{1,\ldots, l\}$ that occurs in $3$ elements of $A$, we call $A$ $3$-value free. Note that this problem is completely characterized by the sequence $A$.

For property testing we call a sequence $A$ $\epsilon$-far from $3$-value free, if more than $\epsilon m$ elements of $A$ have to be changed to establish $3$-value freeness; for that matter it is allowed to assign values that are not yet assigned to one of the elements of $A$, i.e., numbers that are larger than $l$. A property testing algorithm knows the number $m$ of elements and may query the value of the $i$-th element of the sequence for $i=1,\ldots ,m$.

We call an algorithm poisson-$s$ algorithm if it determines randomly by the poisson distribution how many random samples it draws. The analysis in \cite{RRSS09} requires poisson-$s$ algorithms that only get access to the histograms of their random samples; i.e., the algorithms get the information how many values in the sample occur once, twice, thrice, etc., but not the numbers of the elements or their values itself:

\begin{definition}
Let $S$ be a multiset. The histogram $\calH$ of $S$ is a function that assigns to each integer $i\in\NN$ the number of elements of $S$ that occur exactly $i$ times in $S$; i.e.,
$$\calH(i):=\left| \{s\in S| s\text{ is contained in }S\text{ exactly }i\text{ times}\}\right|.$$
\end{definition}

We make use of two Lemmas from \cite{RRSS09} in order to show that for a lower bound for testing $3$-value freeness it suffices to consider poisson-$s$ algorithms that only get access to the histogram of their samples. The following Lemma is a direct corollary of Lemma 3.1\footnote{Cited by Raskhodnikova et al. from \cite{BKS01}.} and of Lemma 5.3 a) and c) from \cite{RRSS09}:

\begin{lemma}\label{lemma_poisalg}
Let $\calA$ be an arbitrary sampling-based algorithm that draws $t$ elements of a sequence $A$; assume that $\calA$ accepts with a probability of at least $\frac{3}{4}$ if $A$ is $3$-value free and rejects with a probability of at least $\frac{3}{4}$ if $A$ is $\epsilon$-far from $3$-value free.

Then there is a poisson-$s$ algorithm $\calA'$ that has an expected number of $s=\Ovon(t)$ queries and only gets access to the histogram of these queries and that accepts with a probability of at least $\frac{3}{4}-o(1)$ if $A$ is $3$-value free and rejects with a probability of at least $\frac{3}{4}-o(1)$ if $A$ is $\epsilon$-far from $3$-value free.
\end{lemma}

Note that the property of $3$-value freeness of a sequence $A$ is closed under permutation of the sequence elements, as is the property of $\epsilon$-farness of $3$-value freeness; thus these properties fulfill the premises of Lemma 3.1 and 5.3 c) in \cite{RRSS09}. The probability guarantees of $\calA'$ follow from those that Lemma 5.3 a) in \cite{RRSS09}, since the statistical distance between the distributions of the results of $\calA$ and $\calA'$ is at most $\frac{4}{s}=o(1)$\footnote{Due to corollary \ref{kor_3starlower1} we can assume $s=\Omega(\sqrt{n})$, and hence it holds $\frac{4}{s}=o(1)$.}.

Analogously to \cite{RRSS09} define a frequency variable $X_A$ as a random variable for the number of occurences in $A$ that a value drawn randomly uniformly distributed from $\{1,\ldots, l\}$ has. The following Lemma is a also a direct corollary from \cite{RRSS09}, namely from Corollary 5.7:

\begin{lemma}\label{lemma_propmom}
Seien $A$ und $B$ Probleminstanzen von $3$-Wert-Freiheit und seien $X_A$ und $X_B$ die zugehörigen Frequenzvariablen. Falls $X_A$ und $X_B$ $k-1$ proportionale Momente haben, dann gilt für jeden Poisson-$s$-Algorithmus $\calA$ mit $s=o(n^{1-1/k})$, der nur die Histogramme seiner Samples sieht
$$\left|\Pr[\calA(A)=\emph{true}] - \Pr[\calA(B)=\emph{true}]\right| = o(1).$$
\end{lemma}

Two random variables $X_1$ and $X_2$ are said to have $k-1$ proportional moments, if $\frac{\E[X_1]}{\E[X_2]} =\frac{\E[X_1^2]}{\E[X_2^2]} = \ldots = \frac{\E[X_1^{k-1}]}{\E[X_2^{k-1}]}$. 

We can now use the preceeding Lemmas in order to prove that any property-testing algorithm for $3$-value freeness needs at least $\Omega(n^{2/3})$ samples.

\begin{lemma}\label{lemma_3wertlowerbound}
Any property-testing algorithm for $3$-value freeness needs at least $\Omega(n^{2/3})$ samples, where $n$ is the length of the input sequence.
\end{lemma}

\begin{proof}
At first note that, due to Lemma \ref{lemma_poisalg}, every property testing algorithm for $3$-value freeness can be replaced by a poisson-$s$ algorithm with an asymptotically equal query complexity, which only gets access to the histogram of its queries -- for sufficiently large sequence length $n$ the additional error probability is neglegible. Hence it suffices to show that every poisson-$s$ algorithm for $3$-value freeness needs to take at least an expected number of $s=\Omega(n^{2/3})$ queries. Furthermore, due to Lemma \ref{lemma_propmom}, every such poisson-$s$-algorithm needs at least $\Omega(n^{2/3})$ queries in expectation to distinguish two sequences $A$ and $B$ whose frequency variables $X_A$ and $X_B$ have to proportional moments. We will now give two such classes, $\calC_A$ being a class of $3$-value free sequences and $\calC_B$ a class of sequences that are $\epsilon$-far from $3$-value free. From the above considerations we can then conclude that testing $3$-value freeness requires $\Omega(n^{2/3})$ queries.

Let  $\calC_A$ be the class of all sequences $A$ of length $n$, such that $n$ is a multiple of $32$ and $A$ contains $\frac{1}{2}n$ distinct values, each of them two times. Every sequence in $A$ is $3$-value free.

Let $\calC_B$ be the class of all sequences $B$ of length $n$, such that $n$ is a multiple of $32$ and the values that occur in $B$ are distributed as follows: There are $\frac{17}{32}n$ distinct values in total, of whom  $\frac{1}{32}n$ values occur $3$ times, $\frac{13}{32}n$ values $2$ times and $\frac{3}{32}n$ values once. For $\epsilon<\frac{1}{32}$, all of the sequences in $\calC_B$ are $\epsilon$-far from $3$-value free.

We will now show that the frequency variables of two sequences of equal length from $\calC_A$ and $\calC_B$ have two proportional moments. Let, for an arbitrary fixed $n$, $A\in\calC_A$ and $B\in\calC_B$ with length $n$ each and let $X_A$ and $X_B$ be the corresponding frequency variables. It holds $\E[X_A] = 2$ and $\E[X_A^2]=4$. Furthermore, we have $$\E[X_B] = \frac{1}{17}\cdot 3 + \frac{13}{17}\cdot 2 + \frac{3}{17}\cdot 1 = \frac{32}{17}$$ and $$\E[X_B] = \frac{1}{17}\cdot 9 + \frac{13}{17}\cdot 4 + \frac{3}{17}\cdot 1 = \frac{64}{17}.$$ We can conclude $$\frac{\E[X_A]}{\E[X_B]} = \frac{2\cdot 17}{32} = \frac{4\cdot 17}{64} = \frac{\E[X_A^2]}{\E[X_B^2]},$$ and hence $X_A$ and $X_B$ have two proportional moments.
\end{proof}

Wir werden nun durch Reduktion zeigen, dass jeder Property-Testing-Algorithmus für $3$-Stern-Freiheit $\Omega(n^{2/3})$ Anfragen benötigt, wenn dies auch eine untere Schranke für das Testen von $3$-Wert-Freiheit ist.

\begin{mtheorem}
In the adjacency list model for directed graphs where algorithms cannot query the incoming edges of a vertex, every property testing algorithm for $3$-star freeness has a query complexity of $\Omega(n^{2/3})$, where $n$ is the number of vertices of the input graph.
\end{mtheorem}

\begin{proof}
For a proof by contradiction we assume that there is a property testing algorithm $\calA$ for $3$-star freeness that has a query complexity of $o(n^{2/3})$ for input graphs with $n$ vertices. We will construct an algorithm $\calA'$ that takes as input a sequence $A$ of length $n'=\frac{1}{2}n$ and a proximity parameter $\epsilon'$, and we will prove that $\calA'$ is a property testing algorithm for $3$-value freeness that has a query complexity of $o(n^{2/3})$.

$\calA'$ works by calling $\calA$ and returning the return value of $\calA$; the proximity parameter $\epsilon$ of $\calA$ is set to $2\epsilon'$ and $\calA$ is given oracle access to a graph $G=(V,E)$ with $n=2n'$ vertices that is dynamically constructed for each query of $\calA$. $G$ consists of isolated stars or single vertices and is defined as follows:
\begin{itemize}
	\item $V$ consists of $n'$ central vertices $v_1, \ldots , v_{n'}$ and $n'$ outer vertices $u_1,\ldots ,u_{n'}$.
	\item For $i=1,\ldots ,n'$, there is an edge from $u_i$ to $v_{A_i}$; i.e., the outer vertices for whose numbers the corresponding elements of $A$ have the same value form a star together with a common central vertex.
	\item $G$ does not contain any further vertices and edges; thus, there may be isolated vertices in $G$.
\end{itemize}

Hence, the vertices $u_i$ represent the elements of the sequence $A$ in $G$, and by having a common central vertex in $G$ it is represented that the corresponding elements of $A$ have the same value.

For a query of $\calA$, $G$ can be constructed locally\footnote{The construction method given here does not necessarily create the above vertex permutation, but the resulting graph is isomorphic to $G$. For the sake of simplicity, we will refer to it as $G$.}: When an unknown vertex is queries, it is determined by throwing a coin whether this vertex will be a central vertex or an outer vertex; the probability for a central vertex is $\frac{n'-c}{2n'-c-o}$, if so far $c$ central vertices and $o$ outer vertices are known to $\calA$. If the queried vertex turns out to be an outer vertex, $\calA'$ queries a sequence element that is randomly uniformly distributed among those that have not already been queried; hence the probability for each of the remaining elements is $\frac{1}{n'-o}$ if, again, $o$ is the number of previously queried outer vertices. Basically this is equivalent to drawing from $A$ without replacement.

Now let $i$ be the number of the sequence element that has been selected in this way: $\calA'$ inserts an edge $(u,v_{A_i})$ into the graph in construction. Since $\calA$ can only query the outgoing edges from vertices and all the edges of $G$ are directed from outer to central vertices, every edge that $\calA$ queries can be constructed in this way.

This construction procedure needs at most $1$ query to $A$ for every query of $\calA$; hence the query complexity of $\calA'$ is at most that of $\calA$, which is $o(n^{2/3}) = o((n')^{2/3})$. Furthermore, each value that occurs $i$ times in $A$ corresponds to an $i$-star in the graph $G$ constructed by $\calA'$. Hence, $G$ is $3$-star free if and only if $A$ is $3$-value free, and, by the choices of $\epsilon$ and $\epsilon'$, $G$ is $\epsilon$-far from $3$-star free if and only if $A$ is $\epsilon'$-far from $3$-value free. Thus, $\calA'$ is a property testing algorithm for $3$-value freeness that has a query complexity of $o((n')^{2/3})$, which is a contradiction to Lemma \ref{lemma_3wertlowerbound}. Hence there does not exist a property testing algorithm $\calA$ for $3$-star freeness that has a query complexity of $o(n^{2/3})$.
\end{proof}

We finally discuss lower bounds for the query complexity of testing $k$-star freeness for $k>3$. The property of $k$-value freeness is defined analogously to $3$-value freeness, and the Lemmas \ref{lemma_poisalg} and \ref{lemma_propmom} hold analogously for $k$-star freeness; the reduction from $k$-value freeness to $k$-star freeness can be done analogously to the special case $k=3$. The difficulty is to give classes $\calC_A$ and $\calC_B$ like in the proof of Lemma \ref{lemma_3wertlowerbound} that have $k-1$ proportional moments. This can be easily done for fixed $k>3$, but it is an open problem to give closed formulae for the frequency distributions of the different $i$-stars in graphs of these classes. We close this chapter with the below conjecture:

\begin{conjecture}
In the adjacency list model for directed graphs where algorithms cannot query the incoming edges of a vertex, every property testing algorithm for $k$-star freeness has a query complexity of $\Omega(n^{1-1/k})$, where $n$ is the number of vertices of the input graph.
\end{conjecture}

%% file: Kap_Connectivity.tex
\section{Testing Strong Connectivity}\label{kap_zusammenhang}

Strong connectivity, i.e., the question whether all pairs of vertices are connected by paths in both directions,  is a very basic property of directed graphs. In this chapter we will give a property testing algorithm for strong connectivity that Property-Testing for this graph property was at first considered by  Bender and Ron \cite{BR02}, who, amongst others, give a lower bound of $\Omega(n^{1/2})$ for any property testing algorithm for this problem under the assumption that the algorithm may only query the incoming edges of a vertex. After first publishing this result at ESA 2012 \cite{HS12} we discovered that there already existed a proof sketch for an algorithm similar to ours, which Oded Goldreich published in the appendix of a survey article about graph property testing \cite{G10}. At the end of this section we will include a detailed discussion about similarities and differences between Goldreich's and our algorithm.


The above-mentioned lower bound construction of Bender and Ron -- more precisely the class of $\epsilon$-far graphs that Bender and Ron give -- can be used for deriving a lower bound of $\Omega(n)$ for any property testing algorithm with \emph{one-sided error} in the same graph model: Let $G$ be a graph with $n$ vertices that consists of a directed circle of $n-\epsilon Dn -1$ vertices, whose edges all have the same direction, and $\epsilon Dn+1$ outer vertices that have exactly one edge to one of the circle vertices; each circle vertex has at most one incoming edge from an outer vertex (see figure \ref{abb_untereschranke}). $G$ is $\epsilon$-far from being strongly connected, since for archieving strong connectivity, every outer vertex needs an additional incoming edge. On the other hand, every subgraph $H$ of $G$ induced by at most $n-\epsilon Dn -2$ vertices can be completed to a strongly connected graph $G'$ with $n$ vertices as follows: There are $\epsilon Dn + 2$ vertices and their adjacent edges to add; for every outer vertex of $G$, add an edge from one of the missing vertices; in particular, since at least one of the circle vertices of $G$ is missing in $H$, this missing circle vertex can be connected to one of the outer vertices. After this, add the remaining edges of $G$ to $G'$, ignoring those that would create an outgoing $2$-star at its source vertex (see again figure \ref{abb_untereschranke}). The resulting graph is strongly connected. Hence, no algorithm that only queries $n-\epsilon Dn-2$ vertices  of $G$ and only its outgoing edges can rule out the possibility that $G$ is strongly connected; hence, every such property testing algorithm for strong connectivity that has one-sided error would have to accept $G$.

\begin{corollary}\label{corollary_connectivityonesidedlowerbound}
Every property testing algorithm for strong connectivity of directed graphs in the adjacency list model with only outgoing edges being visible for the algorithm has a query complexity of $\Omega(n)$.
\end{corollary}

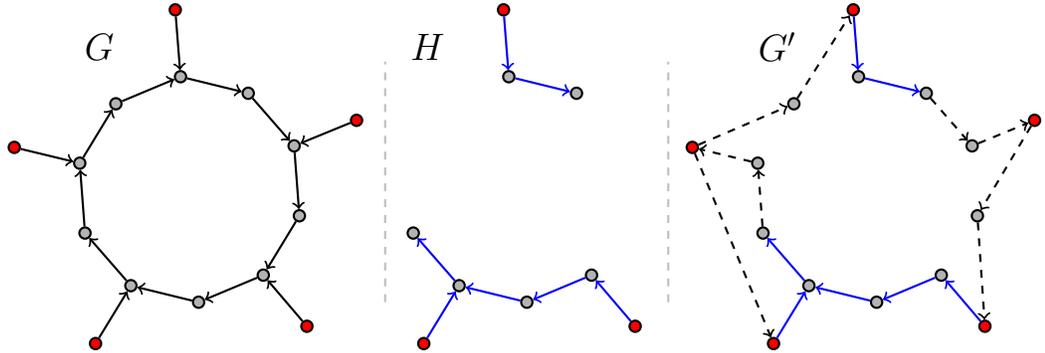
\begin{figure}[tb]
	\centering
	\begin{tikzpicture}[thick]


		\node at (-1.2,1.9) {\Large{$G$}};
		\foreach \i/\x/\y in {1/0.98/-1.14, 2/0.12/-1.50, 3/-0.78/-1.28, 4/-1.39/-0.58, 5/-1.46/0.35, 6/-0.98/1.14, 7/-0.12/1.50, 8/0.78/1.28, 9/1.39/0.58, 10/1.46/-0.35}
			\node (p\i) [draw=black,fill=gray70, minimum size = 0.15cm,inner sep = 0cm,circle] at (\x,\y) {};
		\foreach \i/\j in {1/2,2/3,3/4,4/5,5/6,6/7,7/8,8/9,9/10,10/1}
			\draw[->] (p\i) -- (p\j);

		\foreach \k/\x/\y in {1/1.56/-1.82, 2/-1.25/-2.05, 3/-2.33/0.56, 4/-0.19/2.39, 5/2.22/0.92}
			\node (q\k) [draw=black,fill=red, minimum size = 0.15cm,inner sep = 0cm,circle] at (\x,\y) {};

		\draw[->] (q1) -- (p1); \draw[->] (q2) -- (p3); \draw[->] (q3) -- (p5); \draw[->] (q4) -- (p7); \draw[->] (q5) -- (p9);
		
		\draw[dashed, color=lightgray] (2.6,-1.5) -- (2.6, 1.7);
	\end{tikzpicture}~~
	\begin{tikzpicture}[thick]
		\node at (-1.2,1.9) {\Large{$H$}};
		\foreach \i/\x/\y in {1/0.98/-1.14, 2/0.12/-1.50, 3/-0.78/-1.28, 4/-1.39/-0.58, 7/-0.12/1.50, 8/0.78/1.28}
			\node (p\i) [draw=black,fill=gray70, minimum size = 0.15cm,inner sep = 0cm,circle] at (\x,\y) {};

		\foreach \i/\j in {1/2,2/3,3/4,7/8}
				\draw[->,draw=blue] (p\i) -- (p\j);

		\foreach \k/\x/\y in {1/1.56/-1.82, 2/-1.25/-2.05, 4/-0.19/2.39}
			\node (q\k) [draw=black,fill=red, minimum size = 0.15cm,inner sep = 0cm,circle] at (\x,\y) {};

		\draw[->,draw=blue] (q1) -- (p1); \draw[->,draw=blue] (q2) -- (p3); \draw[->,draw=blue] (q4) -- (p7); 
		\draw[dashed, color=lightgray] (2.0,-1.5) -- (2.0, 1.7);
	\end{tikzpicture}~~
	\begin{tikzpicture}[thick]
		\node at (-1.2,1.9) {\Large{$G'$}};
		\foreach \i/\x/\y in {1/0.98/-1.14, 2/0.12/-1.50, 3/-0.78/-1.28, 4/-1.39/-0.58, 5/-1.46/0.35, 6/-0.98/1.14, 7/-0.12/1.50, 8/0.78/1.28, 9/1.39/0.58, 10/1.46/-0.35}
			\node (p\i) [draw=black,fill=gray70, minimum size = 0.15cm,inner sep = 0cm,circle] at (\x,\y) {};
		\foreach \i/\j in {4/5, 8/9}
			\draw[->, dashed] (p\i) -- (p\j);

		\foreach \i/\j in {1/2,2/3,3/4,7/8}
			\draw[->,draw=blue] (p\i) -- (p\j);

		\foreach \k/\x/\y in {1/1.56/-1.82, 2/-1.25/-2.05, 3/-2.33/0.56, 4/-0.19/2.39, 5/2.22/0.92}
			\node (q\k) [draw=black,fill=red, minimum size = 0.15cm,inner sep = 0cm,circle] at (\x,\y) {};

		\draw[->,draw=blue] (q1) -- (p1); \draw[->,draw=blue] (q2) -- (p3); \draw[->,draw=blue] (q4) -- (p7);  \draw[->, dashed] (q3) -- (p6); \draw[->, dashed] (p9) -- (q5); \draw[->, dashed] (p5) -- (q3); \draw[->, dashed] (p6) -- (q4); \draw[->, dashed] (q5) -- (p10); \draw[->, dashed] (p10) -- (q1); \draw[->, dashed] (q3) -- (q2);
	\end{tikzpicture}

\caption{Graph $G$, example for a subgraph $H$ explored by an algorithm and strongly connected graph $G'$ constructed from $H$}\label{abb_untereschranke}
\end{figure}

In the remainder of this section we will present a property testing algorithm for strong connectivity in directed graphs that has a two-sided error and a query complexity of $\Ovon(n^{1-\epsilon/(3+\alpha)})$; $\alpha>0$ is an arbitrarily small constant. We assume that both the in- and the outdegree of every vertex of the input graph is bounded by a constant $D\in\NN$; hence, the undirected degree of a vertex can at most be $2D$. The total number of edges is still at most $\epsilon Dn$, such that a graph is $\epsilon$-far from another graph, if more than $\epsilon Dn$ adjacency list entries of the first graph have to be modified to obtain the second one.

In the following we can assume $D\geq 2$, since for $D=1$ strong connectivity can be easily tested as follows: If $D=1$, every vertex has at most one incoming edge and one outgoing edge; every weakly connected component in such a graph is either an orientation of a circle or a line-shaped graph, and all the edges in each of these orientations have the same direction. The only strongly connected graph is a circle orientation with all $n$ vertices where all edges have the same direction. Every $\epsilon$-far graph consists of $\Omega(\epsilon n)$ weakly connected components, which means that there are also $\Omega(\epsilon n)$ weakly connected components that have at most $\frac{1}{\epsilon}$ vertices each -- a vertex sample of size $\Theta(1/\epsilon)$ contains one of them with high probability. By the above considerations, starting from an arbitrary vertex of  component, every single of these components can either be completely explored with $\frac{1}{\epsilon}$ queries (in case it is a circle) or a sink vertex can be found within $\frac{1}{\epsilon}$ queries. Hence, for $D=1$, strong connectivity can be tested with a query complexity of $\Theta(\frac{1}{\epsilon^2})$ with one-sided error.

\begin{corollary}
For $D=1$, strong connectivity of a directed graph with both vertex in- and outdegree bounded by $D\in\NN$ can be tested with $\Ovon(\frac{1}{\epsilon^2})$ queries.
\end{corollary}

For larger values of $D$, testing for strong connectivity turns out to be much harder; one main obstacle is that one cannot rely on identifying whitnesses against strong connectivity for $\epsilon$-far graphs: Such whitnesses would either be sink or source components of the input graph, and while it is easy to identify a sink component if the input graph contains many of them, a source component can never be identified without knowing the whole remaining graph, since it cannot be ruled out that an explored area has incoming edges. On the other hand, Bender and Ron have shown that the total number of dead ends in an $\epsilon$-far graph is large:

\begin{lemma}[\cite{BR02}]\label{lemma_deadendzahl}
Let $G$ be a directed graph with $n$ vertices and both its vertex in- and outdegree bounded by $D\in\NN$ and let $\epsilon<1$ be a proximity parameter. If $G$ is $\epsilon$-far from strongly connected, then $G$ has more than $\frac{1}{3}\epsilon Dn$ dead ends.
\end{lemma}

Input graphs that contain many small sink components can be handled in the same way as in the property testing algorithm of Bender and Ron for the model where incomong edges are visible: It suffices to sample $\Ovon(1/\epsilon)$ vertices, explore a small area of size $\Ovon(1/\epsilon)$ vertices around them using breadth first search and reject the input graph, if for one of the sample vertices there are no outgoing edges left during the breadth first search traversal. It can be shown that, if $G$ has many sink components, then $G$ also has many sink components of size $\Ovon(1/\epsilon)$, and hence the above approach will identify one of them with high probability.

\begin{Algorithm}[tb]
\noindent\centering
\shadowbox{
	\begin{minipage}{8cm}
		\begin{tabbing}
			~~~~\= ~~~~\= ~~~~\= ~~~~ \= ~~~~ \= \kill
			\textsc{\textbf{TestSinkFreeness}}($n,G,\beta$)\\
			\> Sample $s_{\ref{alg_senkentest}}=\frac{4}{\beta D}$ vertices of $G$ u.i.d. at random\\
			\> \textbf{Foreach} sampled vertex $v$ {\bf do}\\
			\> \> Start a breadth first traversal at $v$ which stops after exploring $\frac{2}{\beta D}$ vertices\\
			\> \> \textbf{If} the breadth first traverdal completely explores a sink component\\
			\> \> \> \> \textbf{then return false}\\
			\> \textbf{return true}
		\end{tabbing}
	\end{minipage}
}
\caption{{\sc TestSinkFreeness}}\label{alg_senkentest}
\end{Algorithm}

\begin{lemma}\label{lemma_senkentest}
Let $G$ be a directed graph with $n$ vertices both its vertex in- and outdegree bounded by $D\in\NN$ and let $\beta<1$ be a proximity parameter.  If $G$ does not have any sink components, {\sc TestSinkFreeness}($n,G,\beta$) always returns \emph{true}; if $G$ has at least $\beta Dn$ sink components, {\sc TestSinkFreeness}($n,G,\beta$) returns \emph{false} with a probability of at least $1-e^{-2}$. The query complexity and the running time of the algorithm both are $\Ovon\left(\frac{1}{\beta^2 D^2}\right)$.
\end{lemma}

\begin{proof}
If $G$ does not have any sink components, the condition queried by the if statement in the second but last line of {\sc TestSinkFreeness} can never turn out to be true; hence the algorithm alway accepts the input.

Now assume that $G$ contains at least $\beta Dn$ sink components. Since $G$ has $n$ vertices Knoten and each of them belongs to at most one sink component, at least $\frac{1}{2}\beta Dn$ of the sink components have $\frac{2}{\beta D}$ or less vertices. If a vertex $v$ of such a sink component is sampled, then the breadth first search that starts at $v$ completely explores the sink component. Hence the algorithm correctly returns  \emph{false}. The probability that none of the $s_{\ref{alg_senkentest}}$ sample vertices lie in one of the small sink components is at most $(1-\frac{1}{2}\beta D)^{s_{\ref{alg_senkentest}}} \leq e^{-2} $.

Running time and query complexity of the algorithm arise from the number $s_{\ref{alg_senkentest}}$ of sample vertices taken mulitplied by the maximum number of vertices explored by the breadth first search.
\end{proof}

As discussed above, source components that are not sink components at the same time are much harder to identify as sink components. Indeed, corollary \ref{corollary_connectivityonesidedlowerbound} implies that they cannot be directly identified (for example by a breadth first search) in sublinear time. However, there may be graphs that are $\epsilon$-far from strong connectivity and that have very few sink components, and a property testing algorithm for strong connectivity has to identify such graphs. Hence our algorithm will identify the presence of many source components indirectly: We will at first assume that the input graph $G$ has many small source components and that each of them contains of exactly one vertex; we will show that in this case there is statistical evidence for $G$ not being strongly connected. In the next step we will give a reduction that is locally constructible -- that is, given $G$, the reduced graph can be sampled from and explored with a constant number of queries in $G$ for each query to the reduced graph-- and that converts a graph with many constant-size source components into a graph that has equally many source components with a size of one each.

\begin{Algorithm}[tb]
\noindent\centering\shadowbox{
	\begin{minipage}{8cm}
		\begin{tabbing}
			~~~~\= ~~~~\= ~~~~\= ~~~~\= ~~~~ \= \kill
			\textsc{\textbf{EstimateReachableVertices}}($n,G,D,\epsilon$)\\
			\> \textbf{for} $i\leftarrow D$ \textbf{downto} $1$ \textbf{do}\\
			\> \> Sample each edge of $G$ with a probability of $p_i=a_i n^{-1/i}$,\\
			\> \> \textbf{If} more than $t_i=16Da_i n^{1-1/i}$ edges are sampled \textbf{then return} n\\
			\> \> $\hat c_i\leftarrow$ number of vertices for which exactly $i$ incoming edges have been sampled\\
   			\> \> $\hat n_i \leftarrow \frac{\hat c_i}{p_i^i} - \sum_{i<j\leq d}{j\choose i} \hat n_j (1-p_i)^{j-i}$\\
			\> {\bf return} $\hat m = \frac{\epsilon Dn}{16} + \sum_{1\leq i\leq D}\hat n_i$
		\end{tabbing}
	\end{minipage}
}
\caption{EstimateReachableVertices}\label{alg_erreichbareknoten}
\end{Algorithm}

Hence we reduce the problem of testing a graph for strong connectivity to the problem of approximating the number of vertices that have no incoming edge. We solve the latter problem by computing collision statistics for $i$-way collisions of common target vertices of a set of sampled edges, for $i=1,\ldots ,D$; combining these statistics allows us to measure the number of vertices that do not have an incoming edge and thus cannot appear in any of the collision statistics, affecting their outcomes significantly if there are sufficiently many such vertices. At this moment, we only state the following lemma, the proof will be given at the end of the section\footnote{In \cite{RRSS09}, section $4$, footnote $4$, a similar approach for approximating a distribution with several $i$-way collision statistics as in \textsc{EstimateReachableVertices} is mentioned, but its correctness is not proved there.}:

\begin{lemma}\label{lemma_erreichbareknoten}
Let $G$ be a directed graph with $n$ vertices both its vertex in- and outdegree bounded by $D\in\NN$ and let $\epsilon<1$ be a proximity parameter. Let $m$ be the number of vertices of $G$ that have at least one incoming edge. Then, {\sc EstimateReachableVertices}($n,G,\epsilon$) returns an estimate $\hat m$ for which $$m\leq \hat m\leq m+\frac{\epsilon Dn}{8} $$ holds with a probability of at least $\frac{3}{4}$. The algorithm needs at most $\Ovon\left(\frac{D^{8D+15}\log D}{\epsilon^3}n^{1-1/D}\right)$ queries to $G$.
\end{lemma}

The basic proof strategy is to show by induction that, for appropriate choice of the $a_i$, the estimators $\hat n_i$ approximate the number of vertices with exactly $i$ incoming edges with high probability (see algorithm \ref{alg_erreichbareknoten}). Then, by the Union Bound, $\hat m$ satisfies the constraints given in Lemma \ref{lemma_erreichbareknoten} with high probability.

We will now introduce the reduction function $C$. For this purpose, we need the notion of a \emph{compact component}; the reduction will create a graph in which every compact component in the original graph is contracted into a single vertex:

\begin{figure}
	\centering
	\begin{tikzpicture}[thick]
		\node at (-.3,2) {\Large{(a)}};

		\fill[rotate=38, color=blue!20] (.95,-.05) ellipse (1.5cm and 0.99cm);
		\node[color=blue!60] at (-.7,1) {\Large{$U_1$}};

		\fill[rotate=-116.5, color=red!20] (-1.9,2.7) ellipse (0.9cm and 0.33cm);
		\node[color=red!60] at (2.7,1) {\Large{$U_2$}};

		\foreach \i/\x/\y in {1/0/-.5, 2/1/0, 3/0/1, 4/2/1, 5/-1.5/0, 6/2/-1, 7/3/0, 8/3.5/1}
			\node (p\i) [draw=black,fill=gray70, minimum size = 0.15cm,inner sep = 0cm,circle] at (\x,\y) {};


		\foreach \i/\j in {1/2, 2/4, 4/3, 3/5, 3/1, 1/5, 2/6, 1/6, 6/7, 4/7, 7/8, 8/7}
			\draw[->] (p\i) -- (p\j);
	
		\draw[dashed, color=lightgray] (4.2,-1.2) -- (4.2, 1.7);
	\end{tikzpicture}~~~
	\begin{tikzpicture}[thick]
		\node at (-.3,2) {\Large{(b)}};
		\fill[rotate=38, color=blue!20] (.95,-.05) ellipse (1.5cm and 0.99cm);
		\node[color=blue!60] at (-.4,1.1) {\Large{$U$}};

		\foreach \i/\x/\y in {2/1/0, 3/0/1, 5/-1.5/0, 7/3/0}
			\node (p\i) [draw=black,fill=gray70, minimum size = 0.15cm,inner sep = 0cm,circle] at (\x,\y) {};

		\foreach \i/\x/\y in {1/0/-.5, 4/2/1, 6/2/-1}
			\node (p\i) [draw=black,fill=red, minimum size = 0.15cm,inner sep = 0cm,circle] at (\x,\y) {};
		\node[color=black] at (0,-.9) {\Large{$u_2$}};		
		\node[color=black] at (2.45,1) {\Large{$u_1$}};
		\node[color=black] at (2.45,-1) {\Large{$v$}};

		\foreach \i/\j in {1/2, 2/4, 4/3, 3/5, 3/1, 1/5, 6/2}
			\draw[->] (p\i) -- (p\j);
		
		\foreach \i/\j in {6/1, 7/6, 4/7}
			\draw[->, dashed] (p\i) -- (p\j);
	\end{tikzpicture}
	\caption{(a) Compact components $U_1$, $U_2$; (b) for $\frac{3+\alpha}{\epsilon D}= 5$, $U$ is no compact component, since the highlighted path violates the third condition.}\label{abb_kompkomp}
\end{figure}
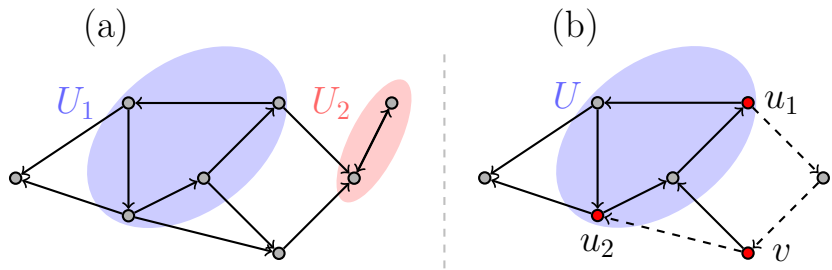

\begin{definition}\label{def_kompaktekomp}
Let $G$ be a directed graph with $n$ vertices both its vertex in- and outdegree bounded by $D\in\NN$ and let $\epsilon<1$ be a proximity parameter; let $\alpha>0$ be a fixed, but arbitrary constant. We call a set of vertices $U\subseteq V$ \emph{compact component} if the following three conditions hold:
\begin{enumerate}
	\item $|U|\leq \frac{3+\alpha}{\epsilon D}$;
	\item The subgraph of $G$ that is induced by $U$ is strongly connected;
	\item There are no vertices $v\in V\backslash U$ and $u_1, u_2\in U$ such that there are paths from $u_1$ tp $v$ and from $v$ to $u_2$, each having a length of at most $\frac{3+\alpha}{\epsilon D}$.
\end{enumerate}
\end{definition}

The following Lemma shows that the compact component of every vertex of a graph is unique:

\begin{lemma}\label{lemma_kompkompeindeutig}
Let $G=(V,E)$, $D$, $\epsilon$ and $\alpha$ be defined as in definition \ref{def_kompaktekomp}. Then, every vertex of $G$ belongs to at most one compact component.
\end{lemma}

\begin{proof}
Let $v\in V$ be an arbitrary vertex. Assume that $v$ belongs to two distinct compact components $U$ and $W$. Since $U\neq W$, there is at least one vertex that is in one of the components, but not in the other; let, without loss of generality, $u\in U\backslash W$ be such a vertex. Since $U$ is a compact component, $U$ is strongly connected and contains at most $\frac{3+\alpha}{\epsilon D}$ vertices. Hence, since $v,u\in U$, there are paths from $v$ to $u$ and from $u$ to $v$, each having a length of at most $\frac{3+\alpha}{\epsilon D}$. This is, however, a contradiction to the assumption that $W$ is a compact component: Since $v\in W$ and $u\notin W$, the existence of the above paths violates the third condition for compact components. Hence there are not two distinct compact components $v$ belongs to.
\end{proof}

There may be vertices of a graph that do not belong to any compact component by the above definition. In the following we will assume that such vertices form their own compact components. Together with the above Lemma this means that every vertex of a graph can be mapped to a unique compact component:

\begin{definition}
Let $G=(V,E)$, $D$, $\epsilon$ and $\alpha$ be defined as in definition \ref{def_kompaktekomp}. Then, for a vertex $v\in V$, $C(v)$ denotes the compact component that $v$ belongs to, if there exists such a component, or $C(v)=\{v\}$ if $v$ does not belong to a compact component.

Let additionally $C(G)$ be the graph that results from contracting every compact component of $G$.
\end{definition}

Next we show that every vertex of a compact component has the same compact component.

\begin{lemma}
Let $G=(V,E)$, $D$, $\epsilon$ and $\alpha$ be defined as in definition \ref{def_kompaktekomp}. Then, $C(u)=C(v)$ for all vertices $v\in V$ and $u\in C(v)$.
\end{lemma}

\begin{proof}
We assume $|C(v)|>1$, since otherwise the Lemma holds trivially because of $C(v)=\{v\}$. Furthermore we assume that there os a vertex $u\in C(v)$ such that $C(u)\neq C(v)$; i.e., there exists a vertex that is only contained in either in $C(v)$ or in $C(u)$.

First assume that there exists $w\in C(u)\backslash C(v)$: Then, analogously to the proof of Lemma \ref{lemma_kompkompeindeutig}, there exist paths from $u$ to $w$ and from $w$ to $u$ of length at most $\frac{3+\alpha}{\epsilon D}$ each, and, since $v\in C(v)$ but $w\notin C(v)$, this violates the third condition for compact components for $C(v)$ -- a contradiction to the assumption that $C(v)$ is a compact component.

Now assume that there exists $w\in C(v)\backslash C(u)$ gibt. Analogously to the first case this implies the existence of short paths from $v$ to $w$ and from $w$ to $v$, contradicting with the third condition for compact components for $C(u)$.
\end{proof}

By the above Lemma, the mapping $C$ is well-defined. The fundamental property of this mapping is that it maps every small source component of a graph to a compact component; in $C(G)$, each of these compact components gets contracted to a single vertex, such that $C(G)$ contains many vertices with an indegree of $0$ if $G$ contains many small source components.

\begin{observation}
Let $G=(V,E)$, $D$, $\epsilon$ and $\alpha$ be defined as in definition \ref{def_kompaktekomp} and let $U\subseteq V$ be a source component that contains at most $\frac{3+\alpha}{\epsilon D}$ vertices. Then, $U$ is a compact component.
\end{observation} 

\begin{proof}
By the assumption of the Lemma, $|U|\leq\frac{3+\alpha}{\epsilon D}$, and by definition of source components, $U$ is strongly connected. Additionally, there is no edge from $V\backslash U$ to $U$, and thus there is not path from any vertex in $V\backslash U$ to a vertex in $U$. Hence the three conditions for compact components are satisfied for $U$.
\end{proof}

For the remainder of this section we assume that $\alpha>0$ is chosen fixed but arbitrary; we will assume an appropriate choice for $\alpha$ implicitly whenever we refer to $C(G)$ for a graph $G$. Analogously, $\epsilon$ is used implicitly there. It remains to show that a compact component of a given vertex can be determined in a constant number of queries and how sampling in a contracted graph $C(G)$ can be realized.

Determining the compact component of a given vertex $v$ can be done with $\Ovon\left(\frac{1}{\epsilon}\cdot D^{\frac{6+2\alpha}{\epsilon D}-1}\right)$ queries as follows: Start a breadth first search at $v$ that explores up to a maximal depth of $\frac{3+\alpha}{\epsilon D}$ in order to find every vertex that can possible be in a compact component that includes $v$. Let $U$ be the maximal set of vertices that have been explored in this way such that $v\in U$ and the subgraph of $G$ induced by $U$ is strongly connected; if $|U|> \frac{3+\alpha}{\epsilon D}$, it holds $C(v)=\{v\}$, elsewise $U$ is the candidate set for the compact component of $v$. We verify $U$ by checking for vertices that violate the third condition for compact components: For this purpose we perform a breadth first search with a maximum depth of $\frac{6+2\alpha}{\epsilon D}$, having $U$ as the set of starting vertices. The depth boundary on this breadth first search traversal makes for the above-mentioned query complexity.

Sampling a vertex in $C(G)$ is realized by at first sampling a vertex of $G$ and then computing the connected component that the vertex belongs to. A vertex sample is accepted with a probability of $1/i$, if $i$ is the number of vertices in the corresponding component; in this case, the vertex of $C(G)$ that represents the sampled vertex of $G$ is returned. Otherwise, the sample is thrown away. This procedure ensures that all vertices have the same sample probability.

For this procedure, several tries may be necessary until a vertex or an edge has been successfully sampled; however, the probability of a successful sample is at least $\frac{\epsilon D}{3+\alpha}$, since this is the inverse of the maximum number of vertices that can be represented by a vertex of $C(G)$; hence, sampling $k$ times yields $k\cdot \frac{\epsilon D}{3+\alpha}$ successful samples in expectation, and by Markov's Inequality, sampling $\frac{k}{p}\cdot \frac{3+\alpha}{\epsilon^2}$ yields $k$ successful samples with a probability of at least $(1-p)$. This leads to the following observation:

\begin{observation}\label{obs_sampelzeit}
Let $G=(V,E)$, $D$, $\epsilon$ and $\alpha$ be defined as in definition \ref{def_kompaktekomp} and let $p\in (0,1)$. Sampling $k$ vertices of $C(G)$ uniformly and independently distributed  and with an error probability of at most $p$ can be done with $\Ovon\left(\frac{k}{\epsilon^2p}\cdot D^{\frac{6+2\alpha}{\epsilon D}-2} \right)$ queries to $G$.
\end{observation}

For using the algorithm \textbf{TestSinkFreeness} on a contracted graph $C(G)$, an estimate for the number of vertices of $C(G)$ is needed. The algorithm {\sc EstimateVertexNumber} returns such an estimate with high probabilty; in contrast to the approximation algorithm for the number of connected components of an undirected graph by Chazelle et al. (\cite{CRT05}), here the idea is to sample vertices of $G$ and increase the estimate by the inverse of the size of the compact component of the sample vertex.

\begin{Algorithm}[tb]
\noindent\centering
\shadowbox{
	\begin{minipage}{8cm}
		\begin{tabbing}
			~~~~\= ~~~~\= ~~~~\= ~~~~ \= \kill
			\textsc{\textbf{EstimateVertexNumber}}($n,G,\epsilon$)\\
			\> $x\leftarrow 0$\\
			\> Sample $s_{\ref{alg_knotenzahl}}=\frac{3}{\epsilon^{2} d^{2}}$ vertices of $G$ u.i.d. at random\\
			\> {\bf Foreach} sampled vertex $v$ {\bf do} $x\leftarrow x+\frac{1}{|C(v)|}$\\
			\> {\bf return} $\hat n = \frac{n}{t}x$
		\end{tabbing}
	\end{minipage}
}
\caption{{\sc EstimateVertexNumber}}\label{alg_knotenzahl}
\end{Algorithm}

\begin{lemma}\label{lemma_knotenzahl}
Let $G$ be a directed graph with $n$ vertices both its vertex in- and outdegree bounded by $D\in\NN$ and let $\epsilon<1$ be a proximity parameter; let $\tilde n$ be te number of vertices of $C(G)$. Then, for the value $\hat n$ that is returned by {\sc EstimateVertexNumber}($n,G,\epsilon$) it holds 
$$\tilde n-\epsilon Dn\leq \hat n \leq \tilde n+\epsilon Dn $$
with a probabilty of at least $1-2e^{-3}$. The algorithm has a query complexity of $\Ovon\left(\frac{1}{\epsilon^3 D^3}\cdot D^{(6+\alpha)/(\epsilon D)}\right)$.
\end{lemma}

\begin{proof}
Let, for $i=1,\ldots , \frac{3+\alpha}{\epsilon D}$, $k_i$ be the number of compact components of $G$ that consist of exactly $i$ vertices. Let $X_j$ be a random variable for the contribution of the $j$-the sampe vertex to $x$, $j=1,\ldots , s_{\ref{alg_knotenzahl}}$, i.e., $X_j=\frac{1}{C(v_j)}$. Let $X=\sum_{1\leq j\leq s_{\ref{alg_knotenzahl}}}X_j$. For every $j$ it holds
$$E[X_j]=\sum_{1\leq i\leq \frac{3+\alpha}{\epsilon D}}\frac{1}{i}\cdot \Pr[C(v_j)=i] = \sum_{1\leq i\leq \frac{3+\alpha}{\epsilon D}}\frac{1}{i}\cdot \frac{k_i\cdot i}{n} = \frac{1}{n}\cdot \sum_{1\leq i\leq \frac{3+\alpha}{\epsilon D}} k_i = \frac{\tilde n}{n},$$
since the sum in the first but last step is the number of compact components of $G$ and hence the number of vertices of $C(G)$. Moreover,
$$\E[\hat n] = \frac{n}{s_{\ref{alg_knotenzahl}}}\cdot\E[X] = \frac{n}{s_{\ref{alg_knotenzahl}}}\cdot\sum_{1\leq j\leq s_{\ref{alg_knotenzahl}}}\E[X_j] =  \frac{n}{s_{\ref{alg_knotenzahl}}}\cdot\sum_{1\leq j\leq s_{\ref{alg_knotenzahl}}}\frac{\tilde n}{n} = \tilde n,$$
and since all $X_j$ are between $0$ and $1$, applying the Hoeffding Inequality yields
$$\Pr[|\hat n - \E[\hat n]|>\epsilon dn] = \Pr[|X - \E[X]|>\epsilon ds_{\ref{alg_knotenzahl}}] \leq 2\exp\left(-\frac{\epsilon^2 d^2 s_{\ref{alg_knotenzahl}}^2}{\sum_{1\leq i\leq s_{\ref{alg_knotenzahl}}}(1-0)^2}\right) = 2e^{-3}.$$

By the above considerations, determining the compact component of a vertex of $G$ takes at most $\Ovon\left(\frac{1}{\epsilon D}\cdot D^{(6+\alpha)/(\epsilon D)}\right)$ queries to $G$; since this is done for  $s_{\ref{alg_knotenzahl}}=\frac{3}{\epsilon^{2} d^{2}}$ vertices in the algorithm, the overall query complexity is $\Ovon\left(\frac{1}{\epsilon^3 D^3}\cdot D^{(6+\alpha)/(\epsilon D)}\right)$.
\end{proof}

Now we can state our property testing algorithm {\sc TestStrongConnectivity}; it at first calls {\sc TestSinkFreeness} in order to test whether there are many sink components; if this is not thr case, the algorithm estimates the number of vertices of $C(G)$ by calling {\sc EstimateVertexNumber} and uses this value to measure the number of vertices that have at least one incoming edge by calling {\sc EstimateReachableVertices}. If the estimate returned by the latter call is much smaller than the estimated number of vertices of $C(G)$, then {\sc TestStrongConnectivity} rejects; elsewise it accepts.

\begin{Algorithm}[tb]
	\centering
\shadowbox{
	\begin{minipage}{8cm}
		\begin{tabbing}
			~~~~\= ~~~~\= ~~~~\= ~~~~ \= \kill
			\textsc{\textbf{TestStrongConnectivity}}($n,G,\epsilon$)\\
			\> {\bf if} $\epsilon D\geq 3+\alpha$ {\bf then return true}\\
			\> \textbf{if not}  {\sc TestSinkFreeness}($n,G,\frac{\alpha\epsilon}{6(3+\alpha)}$) \textbf{then return} false\\
			\> $\hat n \leftarrow ${\sc EstimateVertexNumber}$(C(G),\frac{\alpha}{24(3+\alpha)}\epsilon)$\\
			\> $\hat m \leftarrow ${\sc EstimateReachableVertices}$(n,C(G),\frac{3+\alpha}{\epsilon},\frac{\alpha\epsilon^2 D}{6(3+\alpha)^2})$\\
			\> {\bf if} $\tilde m<\hat n-\frac{\alpha}{12(3+\alpha)}\epsilon Dn$ {\bf then return false}\\
			\> {\bf else return true}
		\end{tabbing}
	\end{minipage}
	\caption{\textsc{TestStrongConnectivity}}
}
\end{Algorithm}

\begin{mtheorem}\label{theorem_zusammenhang}
Let $G$ be a directed graph with $n$ vertices both its vertex in- and outdegree bounded by $D\in\NN$ and let $\epsilon<1$ be a proximity parameter. Then, {\sc TestStrongConnectivity}($n,G,\epsilon$) returns \emph{false} with a probability of at least $\frac{2}{3}$, if $G$ is $\epsilon$-far from strongly connected; it returns  \emph{true} with a probability of at least $\frac{2}{3}$, if $G$ is strongly connected. The query complexity of the algorithm is $\Ovon\left(\alpha^{-4}\left(\frac{4}{\epsilon}\right)^{32/\epsilon+23}D^{144/(\alpha\epsilon^2)}n^{1-\epsilon/(3+\alpha)} \log\frac{1}{\epsilon}\right)$.
\end{mtheorem}

\begin{proof}
Let $\tilde n$ be the number of vertices of $C(G)$ and let $\tilde m$ be the number of those vertices that have at least one incoming edge. Let $\tilde D$ be the maximum in- or outdegree of a vertex in $C(G)$: Because of the maximum size of a compact component it holds $\tilde D\leq \frac{3+\alpha}{\epsilon}$.

At first assume that $G$ is $\epsilon$-far from strongly connected. We will show that then the probability that {\sc TestStrongConnectivity} inadvertantly returns \emph{true} is bounded by a constant. Initially note that this cannot happen in the first line of the algorithm: For $\epsilon\geq 3+\alpha$ there is no graph that is  $\epsilon$-far from strongly connectes, since in this case it would be possible to remove all edges of the graph and insert a circle over all vertices instead, not surpassing $\epsilon Dn$ edge modifications; hence, for an $\epsilon$-far graph the condition of the {\bf if}-statement in line 1 will never hold true.

Since $G$ is $\epsilon$-far from strongly connected, $G$ contains more than $\frac{1}{3}\epsilon Dn$ dead ends by Lemma \ref{lemma_deadendzahl}. Hence it contains at least $\frac{\alpha}{6(3+\alpha)}\epsilon Dn$ sink components or at least $\frac{1}{3}\epsilon Dn - \frac{\alpha}{6(3+\alpha)}\epsilon Dn = \frac{6+\alpha}{6(3+\alpha)}\epsilon Dn$ source components. In the first case Lemma \ref{lemma_senkentest} guarantees that {\sc TesteSenkenFeiheit} returns \emph{true} with a probability of at most $e^{-2}<\frac{1}{3}$ \emph{true}.

In the second case at least $\frac{\alpha}{6(3+\alpha)}\epsilon Dn$ of the source components have a size of at most $\frac{3+\alpha}{\epsilon D}$, since they are vertex-disjoint and together contain at most $n$ vertices. Thus the number of vertices without incoming edges is at least $\frac{\alpha}{6(3+\alpha)}\epsilon Dn$, too; i.e., it holds $\tilde m\leq  \tilde n - \frac{\alpha}{6(3+\alpha)}\epsilon Dn$.

We can use this bound on the number of non-reachable vertices of $C(G)$ to gain an upper bound for $\hat m$ that holds with high probability: Lemma \ref{lemma_erreichbareknoten} guarantees that $\hat m\leq \tilde m +\frac{1}{8}\cdot\frac{\alpha\epsilon^2 D}{6(3+\alpha)^2}\cdot \tilde D n$ holds with a probability of at least $\frac{3}{4}$; moreover, by Lemma \ref{lemma_knotenzahl} we have $\hat n\geq \tilde n - \frac{\alpha}{24(3+\alpha)}\epsilon Dn$ with a probability of at least $1-2e^{-3}$. These two facts combined with the bound of $\tilde D\leq D\cdot\frac{3+\alpha}{\epsilon D} = \frac{3+\alpha}{\epsilon}$ for the in- and outdegree of a vertex of $C(G)$ yield
\begin{align*}
	\hat m & \leq \tilde m +\frac{1}{8}\cdot\frac{\alpha\epsilon^2 D}{6(3+\alpha)^2}\cdot \tilde D n \leq \tilde n - \frac{\alpha}{6(3+\alpha)}\epsilon Dn + \frac{1}{8}\cdot\frac{\alpha}{6(3+\alpha)}\epsilon Dn\\
	& \leq \hat n +  \frac{1}{4}\cdot \frac{\alpha}{6(3+\alpha)}\epsilon Dn - \frac{7}{8}\cdot \frac{\alpha}{6(3+\alpha)}\epsilon Dn < \hat n - \frac{\alpha}{12(3+\alpha)}\epsilon Dn.
\end{align*}
Hence, by the Union Bound, {\sc TestStrongConnectivity}($n,G,D,\epsilon$) returns \emph{false} with a probability of at least $1-\frac{3}{4}-2e^{-3}>\frac{2}{3}$ in this case.

Finally we consider the case that $G$ is strongly connected, i.e., there are no dead ends in $G$. By Lemma \ref{lemma_senkentest}, the call to {\sc Teste\-Senken\-Frei\-heit} will hence return true \emph{true}. Particularly, there are no source components in $G$, and thus there are no vertices without incoming edges in $C(G)$ -- every such vertex in $C(G)$ would correspond to a source component in $G$. Hence it holds $\tilde m=\tilde n$. Moreover, Lemma \ref{lemma_erreichbareknoten} guarantees that $\hat m\geq \tilde m$ holds with a probability of at least $\frac{3}{4}$ and by Lemma \ref{lemma_knotenzahl} we have $\hat n\leq \tilde n + \frac{\alpha}{24(3+\alpha)}\epsilon Dn$ with a probability of at least $1-2e^{-3}$. Altogether we can conclude $$\hat m \geq \tilde m = \tilde n \geq \hat n - \frac{\alpha}{24(3+\alpha)}\epsilon Dn>\hat n - \frac{\alpha}{12(3+\alpha)}\epsilon Dn,$$
and hence by the Union Bound {\sc TestStrongConnectivity}($n,G,D,\epsilon$) returns \emph{true} in this case with a probability of at least $\frac{2}{3}$.

The query complexity of the algorihtm is dominated by the query complexity of the call {\sc EstimateReachableVertices}$(n,C(G),\frac{3+\alpha}{\epsilon},\frac{\alpha\epsilon^2 D}{6(3+\alpha)^2}$): Since the degree bound of $C(G)$ is $\frac{3+\alpha}{\epsilon}$ von $C(G)$, for $\alpha<1$ this call queries $C(G)$
\begin{align*}\Ovon & \left(\left(\frac{\alpha\epsilon^2 D}{6(3+\alpha)^2}\right)^{-3}\left(\frac{3+\alpha}{\epsilon}\right)^{8(3+\alpha)/\epsilon+15}n^{1-\epsilon/(3+\alpha)}\log \frac{3+\alpha}{\epsilon}\right)\\ & =\Ovon\left(\alpha^{-3}\left(\frac{4}{\epsilon}\right)^{32/\epsilon+21}n^{1-\epsilon/(3+\alpha)} \log\frac{1}{\epsilon}\right)
\end{align*}
times. Because each of these queries induces $$\Ovon\left(\frac{6(3+\alpha)^2}{\alpha\epsilon^2 D^2}\cdot D^{6(6+\alpha)(3+\alpha)^2/(\alpha\epsilon^2 D^2)}\right)=\Ovon\left(\epsilon^{-2}\alpha^{-1} D^{144/(\alpha\epsilon^2)} \right)$$ queries to $G$, the overall query complexity is
$$\Ovon\left(\alpha^{-4}\left(\frac{4}{\epsilon}\right)^{32/\epsilon+23}D^{144/(\alpha\epsilon^2)}n^{1-\epsilon/(3+\alpha)} \log\frac{1}{\epsilon}\right).$$
\end{proof}

It remains to show \ref{lemma_erreichbareknoten}, i.e., that the estimate $\hat m$ returned by algorithm {\sc EstimateReachableVertices} satisfies $\tilde m\leq \hat m \leq \tilde m + \frac{1}{8}\epsilon Dn$ with a probability of at least $\frac{3}{4}$, if $\tilde m$ is the number of reachable vertices in the input graph $C(G)$, and that this algorithm has a query complexity of $\Ovon(\epsilon^{-3}D^{8D+15}n^{1-1/D}\log D)$.

\begin{proof}[Beweis (Lemma \ref{lemma_erreichbareknoten}).]
We will show by induction that the basic estimators $\hat n_i$ are a good approximation for the number $n_i$ of vertices that have exactly $i$ incoming edges; we then conclude that $\hat m$ is a good estimate for $\tilde m$. At first note that, by Markov's Inequality, the probability that the algorithm aborts in the third line (and possibly returns the wrong value) is at most $\frac{1}{16D}$; the probability that this happens in at least one of the $D$ iterations of the {\bf for}-loop is at most $\frac{1}{16}$ by the union bound hence. If the input graph $G$ is a contracted graph, we can set the error probability in Observation \ref{obs_sampelzeit} to $\frac{1}{16D}$; the probability that in at least one iteration of the {\bf for}-loop the required number of samples is not attained is at most $\frac{1}{16}$ then because of the union bound. For the remainder of this proof we will assume that neither of these events occurs.

Now let $a_i:=a^{D/i}$ für $i=1,\ldots, D$; we will determine the exact value of the constant $a$ later. For all these values for $i$ let $V_i$ be the set of vertices of $G$ that have exactly $i$ incoming edges; we have $n_i=|V_i|$ hence.

We define the random variable $Y_{i,j}$ as the number of  $i$-way collisions on vertices in $V_j$ in the iteration of the {\bf for}-loop that has the loop counter $i$; an $i$-way collision is the event that for a vertex $v$ exactly $i$ of its incoming edges are contained in the edge sample taken in the second line of the algorithm. Moreover let $X_{v,i}$ be an indicator random variable for the event that in the iteration that has the loop counter $i$ the vertex $v$ has an $i$-way collision. For every vertex $v\in V_j$ it holds
\begin{align*}
\E[X_{v,i}] & =\Pr[X_{v,i}=1] = {j \choose i}p_i^i(1-p_i)^{j-i}\\
		&=  {j \choose i}(a_in^{-1/i})^i(1-p_i)^{j-i} = \frac{1}{n}{j \choose i}a^D(1-p_i)^{j-i},
\end{align*}
since the events for the incoming edges of $v$ being in the sample set are independent from each other. Thus the number of incoming edges of $v$ that are contained in the sample set is binomially distributed. Hence we can conclude
$$\E[Y_{i,j}]=\sum_{v\in V_j}\E[X_{v,i}]=\frac{n_j}{n}{j \choose i}a^D(1-p_i)^{j-i}.$$
We can now bound the probability that $Y_{i,j}$ deviates from its expectation by too much. The maximum deviation that we want to allow is $\delta_{i}:=\frac{\epsilon a^D}{2^{i+4}D^{2i-1}}$. Since a Chernoff Bound can only be effectively used if the expected value is relatively large, we distinguish two cases and use a Chernoff Bound, if $n_j$ is relatively large, and Markov's Inequality, if $n_j$ is small.

At first consider the latter case and assume $n_j\leq \frac{\epsilon n(1-p)^{j-i}}{{j\choose i}2^{i+7}D^{2i+1}}$, that is, it holds $\E[Y_{i,j}]\leq \frac{\epsilon a^D(1-p_i)^{2j-2i}}{2^{i+7}D^{2i+1}}$ and hence in particular $\E[Y_{i,j}]<\delta_i$. The latter fact implies that $Y_{i,j}$ can deviate from its expectation by more than $\delta_i$ only by exceeding it. Hence we can conclude
$$\Pr[|Y_{i,j}-\E[Y_{i,j}]|>\delta_i] = \Pr[Y_{i,j}-\E[Y_{i,j}]>\delta_i] < \Pr[Y_{i,j}>\delta_i]\leq \frac{\E[Y_{i,j}]}{\delta_i}\leq \frac{1}{8D^2},$$
applying Markov's Inequality.

Now consider the first of the two cases, i.e., $n_j>\frac{\epsilon n(1-p)^{j-i}}{{j\choose i}2^{i+7}D^{2i+1}}$; this means that $\E[Y_{i,j}]>\frac{\epsilon a^D(1-p_i)^{2j-2i}}{2^{i+7}D^{2i+1}}$. Since $n_j\leq n$, we additionally have $\E[Y_{i,j}]\leq {j \choose i}a^D(1-p_i)^{j-i}\leq \frac{2^{i+4}D^{2i-1}{j\choose i}(1-p)^{j-i}}{\epsilon} \cdot\delta_i$. Since the events $X_{v,i}=1$ are independent for all vertices $v$, we can conclude
\begin{align*}
\Pr[|Y_{i,j}-\E[Y_{i,j}]|>\delta_i] & \leq \Pr\left[\left| Y_{i,j}-\E[Y_{i,j}]\right| > \frac{\epsilon}{2^{i+4}D^{2i-1}{j\choose i}(1-p)^{j-i}}\E[Y_{i,j}]\right]\\
	&\leq 2\exp\left( - \frac{\epsilon^2}{3\cdot 2^{2i+8}D^{4i-2}{j\choose i}^2(1-p)^{2j-2i}}\E[Y_{i,j}] \right)\\
	&< 2\exp\left( - \frac{\epsilon^3a^D}{3\cdot 2^{3i+15}D^{6i-1}{j\choose i}^2} \right)
\end{align*}
by a multiplicatice Chernoff Bound; hence, for sufficiently large $a=\Ovon\left(\frac{D^{8+13/D}}{\epsilon^{3/D}}\log^{1/D} D\right)$ this probability does not exceed $\frac{1}{8D^2}$.

Basically we are only interested in $i$-way collisions that occur on vertices of $V_i$, since this is the set whose size we want to estimate in the iteration of the {\bf for}-loop that has the loop counter $i$. However, for vertices with a degree if $j>i$, $i$-way collisions can also happen, which can distort the estimate for $n_i$. Hence we have to estimate the number of $i$-way collisions $Y_{i,j}$ on vertices in $V_j$, $j>i$, that will happen in the iteration of the {\bf for}-loop that has a loop counter of $i$. We can then subtract these estimates from the number of collisions measured in this iteration and in this way gain an estimator for $n_i$. For estimating the $Y_{i,j}$ we can simply use the estimators $n_j$ from the previous iterations of the {\bf for}-loop.

In what follows we will assume that $|Y_{i,j}-\E[Y_{i,j}]|\leq \delta_i$ holds for $i=1,\ldots D$ and $j\geq i$; by the Union Bound, the probability for this to happen is at least $\frac{7}{8}$, since the probability that a single $Y_{i,j}$ deviates from its expectation by more than $\delta_i$ is, as shown above, at most $\frac{1}{8D^2}$ ist. We will now show by induction over $i$ that, under these assumptions, $\hat n_i$ deviates from $n_i$ by at most $\frac{\epsilon n}{2^{i+3}D^{2i-2}}$.

For the base case we chose the iteration of the {\it for}-loop that has $i=D$ as the loop counter. It holds $$\E[\hat n_D]=\frac{\E[\hat c_D]}{p_D^{D}}=\frac{n}{a_D^{D}}\E[Y_{D,D}]=\frac{n}{a^{D}}\cdot \frac{n_D}{n}{D \choose D}a^D(1-p_D)^{D-D}=n_D.$$ Moreover we have assumed that $Y_{D,D}$ deviates from its expectation by at most $\delta_D=\frac{\epsilon a^D}{2^{D+4}D^{2D-1}}$. Hence $\hat n_D$ deviates by at most $\frac{\delta_Dn}{a^D}=\frac{\epsilon n}{2^{D+4}D^{2D-1}}<\frac{\epsilon n}{2^{D+3}D^{2D-2}}$ from its expectation $\E[\hat n_D]=n_D$.

Now let $i<D$ be an arbitrary value of the loop counter. As induction hypothesis assume that, for all $j>i$, $\hat n_j$ deviates from $n_j$ by at most $\frac{\epsilon n}{2^{j+3}D^{2j-2}}$. Since $\E[Y_{i,j}]=\frac{n_j}{n}{j \choose i}a^D(1-p_i)^{j-i}$ holds for all $j>i$, we can conclude
\begin{align*}
	\E[\hat n_i] = & \E\left[ \frac{\hat c_i}{p_i^i} - \sum_{i<j\leq d}{j\choose i} \hat n_j (1-p_i)^{j-i} \right] = \frac{\E[\hat c_i]}{p_i^i} - \sum_{i<j\leq d}{j\choose i} \E[\hat n_j] (1-p_i)^{j-i} \\ = & \frac{\E[\hat c_i]}{p_i^i} - \sum_{i<j\leq d}{j\choose i} n_j (1-p_i)^{j-i} = \frac{\E[\hat c_i]}{p_i^i} - \sum_{i<j\leq d}\frac{n}{a^D}\E[Y_{i,j}]\\ = & \frac{1}{p_i^i}\left( \E[\hat c_i]-\sum_{i<j\leq d}E[Y_{i,j}] \right) = \frac{E[Y_{i,i}]}{p_i^i} = n_i.
\end{align*}

It remains to bound the maximum deviation of $\hat n_i$ from its expected value. There are two types of deviation that can occur: The first one is that, for $j>i$, $\hat n_j$ may deviate from $\E[\hat n_j]$; due to the induction hypothesis, this deviation is at most $\frac{\epsilon n}{2^{j+3}D^{2j-2}}$, and since $\hat n_j$ gets multiplied by ${j\choose i}(1-p_i)^{j-i}$ when computing $\hat n_i$, the contribution of $\hat n_j$ to the overall deviation is at most $$\alpha_j:={j\choose i}\cdot\frac{\epsilon n}{2^{j+3}D^{2j-2}}\leq D^{j-i}\cdot\frac{\epsilon n}{2^{j+3}D^{2j-2}}=\frac{\epsilon n}{2^{j+3}D^{j+i-2}}\leq \frac{\epsilon n}{2^{i+4}D^{2i-1}}.$$
The second type of error results from the deviation of at most $\delta_i$ that, for $j\geq i$, may occur between $Y_{i,j}$ and $\E[Y_{i,j}]$; due to this deviation, $\hat c_i$ may deviate from $\E[\hat c_i]$. The contribution to the overall deviation is at most
$$\beta_j:=\frac{\delta_i}{p_i^i}=\frac{n}{a^D}\cdot\frac{\epsilon a^D}{2^{i+4}D^{2i-1}} = \frac{\epsilon n}{2^{i+4}D^{2i-1}}.$$
Adding the contributions of both types of deviation, we can bound the overall deviation by
$$|\hat n_i-n_i| = |\hat n_i-\E[\hat n_i]| \leq \sum_{i< j\leq D}\alpha_j + \sum_{i\leq j\leq D}\beta_j < \sum_{i\leq j\leq D}\alpha_j + \beta_j \leq \frac{2D\epsilon n}{2^{i+4}D^{2i-1}} = \frac{\epsilon n}{2^{i+3}D^{2i-2}},$$
which completes the induction step.

Altogether we therefor have
$$\sum_{1\leq i\leq D}|\hat n_i-n_i| \leq \sum_{1\leq i\leq D}\frac{\epsilon n}{2^{i+3}D^{2i-2}}\leq \sum_{1\leq i\leq D}\frac{\epsilon n}{2^{4}} = \frac{\epsilon}{16} Dn,$$
and hence $\hat m$ deviates from $\tilde m+\frac{\epsilon}{16}Dn$ by at most $\frac{\epsilon}{16} Dn$, if the above assumptions on the tightness of the $Y_{i,j}$ hold, i.e.,
$$\tilde m \leq \hat m \leq \tilde m + \frac{\epsilon}{8}Dn.$$

The success probability of the algorithm results from the probabilities for successfully sampling and for all $Y_{i,j}$ being tight to their expected values; by the Union Bound this leads to a success probability of at least $1-\left(\frac{1}{16}+\frac{1}{16}+\frac{1}{8}\right) = \frac{3}{4}$.

The query complexity of the algorithm results from the number of edges that are sampled; since these are at most $\Ovon(Da_in^{1-1/i}) = \Ovon(\epsilon^{-3}D^{8D+14}n^{1-1/i}\log D)$ in the iteration of the {\bf for}-loop that has the loop counter  $i$, the overall number of samples is at most $\Ovon(\epsilon^{-3}D^{8D+15}n^{1-1/D}\log D)$.
\end{proof}

At the end of the section we give a brief overview over the proof sketch for a property testing algorithm for strong connectivity that Oded Goldreich published earlier and independently from our result  (\cite{G10}, Appendix of the  survey article \emph{Introduction to Testing Graph Properties}). Goldreich at first notes that the graph classes that Bender and Ron use for their lower bound of $\Omega(\sqrt{n})$ on testing strong connectivity \cite{BR02} also yield a lower bound of $\Omega(n)$ for testing strong connectivity with a one-sided error; our proof at the beginning of this section also makes use of this observation, considering only the class of $\epsilon$-far graphs Bender and Ron give.

For testing strong connectivity, Goldreich's approach is the same as ours, except for his choice of a reduction function: In particular, he also proposes to at first solve the problem for graphs in which all source components have size $1$ by computing  $i$-way collision statistics for the target vertices of sets of sampled edges for $i=1,\ldots, D$.

Goldreich's reduction function then has a similar effect than ours: Goldreich considers directed circles with a length of at most $s=\left\lceil\frac{4}{\epsilon D}\right\rceil$, and, for each vertex $v$, the set $C_v$ of all those vertices that lie on such a circle together with $v$. If $C_u =C_v$ for all vertices $u\in C_v$, then $C_v$ is contracted. This approach guarantees that all sufficiently small source components are contracted to a single vertex in the resulting graph, and the reduction function can be computed locally.

For Goldreich's reduction, the maximal number of vertices in a contracted component is at most $s^2$ and hence the maximum indegree of a vertex in the resulting graph is $t=s^2\cdot D \approx \frac{16}{\epsilon^2D}$. This leads to a slightly worse query complexity of $\Ovon(n^{1-1/t})\approx \Ovon(1-\frac{\epsilon^2D}{16})$. Goldreich only gives a brief sketch of his result, leaving open any technical details.

%% file: bibfile.tex
\newcommand{\Proc}{In: Proc. of the\ } 
\newcommand{\Journalon}{Journal on\ } 
\newcommand{\Symp}{Symp.\ } 
\newcommand{\Found}{Foundations\ } 
\newcommand{\Techn}{Technology } 
\newcommand{\CompSci}{Computer Science\ } 
\newcommand{\Theor}{Theoretical\ } 
\newcommand{\Ann}{}
\newcommand{\Alg}{Algorithms\ } 
\newcommand{\CompGeom}{Computational Geometry} 
\newcommand{\ALENEX}{Workshop on Algorithm Engineering and Experiments (ALENEX)}
\newcommand{\BEATCS}{Bulletin of the European Association for \Theor \CompSci (BEATCS)}
\newcommand{\CCCG}{Canadian Conference on \CompGeom (CCCG)}
\newcommand{\CIAC}{Italian Conference on \Alg and Complexity (CIAC)}
\newcommand{\COCOON}{\Ann International Computing Combinatorics Conference (COCOON)}
\newcommand{\COLT}{\Ann Conference on Learning Theory (COLT)}
\newcommand{\COMPGEOM}{\Ann ACM \Symp on \CompGeom}
\newcommand{\DCGEOM}{Discr. \& \CompGeom}
\newcommand{\ECCC}{Electronic Colloquium on Computational Complexity (ECCC)}
\newcommand{\ESA}{\Ann European \Symp on \Alg (ESA)}
\newcommand{\FOCS}{IEEE \Symp on \Found of \CompSci (FOCS)}
\newcommand{\FSTTCS}{\Found of Software \Techn and \Theor \CompSci (FSTTCS)}
\newcommand{\ICALP}{\Ann International Colloquium on Automata, Languages and Programming (ICALP)}
\newcommand{\ICCCN}{IEEE International Conference on Computer Communications and Networks (ICCCN)}
\newcommand{\ICDCS}{International Conference on Distributed Computing Systems (ICDCS)}
\newcommand{\IJCGA}{International J. of \CompGeom\ and Applications}
\newcommand{\INFCTRL}{Information and Computation}
\newcommand{\INFOCOM}{IEEE INFOCOM}
\newcommand{\IPCO}{International Integer Programming and Combinatorial Optimization Conference (IPCO)}
\newcommand{\ISAAC}{International \Symp on \Alg and Computation (ISAAC)}
\newcommand{\ISTCS}{Israel \Symp on \TheoryOfComp and Systems (ISTCS)}
\newcommand{\JACM}{J. of the ACM}
\newcommand{\JALGORITHMS}{J. of \Alg}
\newcommand{\JCSS}{J. of Computer and System Sciences}
\newcommand{\LNCS}{Lecture Notes in \CompSci}
\newcommand{\PODS}{ACM SIGMOD \Symp on Principles of Database Systems (PODS)}
\newcommand{\RANDOM}{International Workshop on Randomization and Approximation Techniques in \CompSci (RANDOM-APPROX)}
\newcommand{\RSA}{Random Structures and Alg.}
\newcommand{\SICOMP}{SIAM \Journalon Comp.}
\newcommand{\SIJDM}{SIAM \Journalon Discrete Mathematics}
\newcommand{\SoCG}{\COMPGEOM}
\newcommand{\SODA}{\Ann ACM-SIAM \Symp on Discr. \Alg (SODA)}
\newcommand{\SPAA}{\Ann ACM \Symp on Parallel \Alg and Architectures (SPAA)}
\newcommand{\STACS}{\Ann \Symp on \Theor Aspects of \CompSci (STACS)}
\newcommand{\STOC}{\Ann ACM \Symp on the Theory of Computing (STOC)}
\newcommand{\SWAT}{Scandinavian Workshop on Algorithm Theory (SWAT)}
\newcommand{\TALG}{ACM Transactions on \Alg}
\newcommand{\TCS}{\Theor \CompSci}
\newcommand{\UAI}{Conference on Uncertainty in Artificial Intelligence (UAI)}
\newcommand{\WADS}{Workshop on \Alg and Data Structures (WADS)}